\documentclass{jocg}

\usepackage[all]{hypcap}
\usepackage{amssymb}
\usepackage{amsmath}
\usepackage{amsthm}
\usepackage{mathtools}

\usepackage{graphicx}

\usepackage{xspace}
\usepackage{paralist} 
\usepackage{subfigure}

\usepackage[english]{babel}
\usepackage[T1]{fontenc}
\usepackage{tikz}
\usetikzlibrary{arrows,calc,decorations.text}

\usepackage{todonotes}

\usepackage[plain,ruled]{algorithm}
\usepackage[noend]{algorithmic}

\usepackage{fullpage}
\usepackage{lineno}

\newtheorem{theorem}{Theorem}
\newtheorem{problem}[theorem]{Problem}
\newtheorem{lemma}[theorem]{Lemma}
\newtheorem{cor}[theorem]{Corollary}

\newtheorem{theorem*}{Theorem}
\newcommand{\remove}[1]{}
\newcommand{\rephrase}[2]{\ \newline \ \noindent\textbf{#1}.~\emph{#2}\newline }

\newcommand{\lr}[1]{\langle #1 \rangle}
\newcommand{\mlr}[1]{\ensuremath{\lr{#1}}\xspace}

\newcommand{\morph}[1]{\mlr{#1}}
\newcommand{\mmorph}[1]{$\mlr{#1}$\xspace}

\newcommand{\cand}{candidate\xspace}
\newcommand{\cv}{\cand vertex\xspace}

\newcommand{\xc}{$x$-contractible\xspace}

\newcommand{\face}[1]{$(#1) $\xspace}

\usetikzlibrary{decorations.pathmorphing}

\tikzstyle{graphnode}=[circle, draw, fill=black!30, inner sep=0pt, minimum
	    width=4pt]
\tikzstyle{unmatched}=[graphnode,fill=black!0]
\tikzstyle{shaded}=[graphnode,fill=black!20]
\tikzstyle{matched}=[graphnode,fill=black!100]
\tikzstyle{matching} = [ultra thick]

\newcommand{\rednote}[1]{\textcolor{red}{{#1}}}

\newcommand{\Deg}{\textrm{deg}}
\newcommand{\MM}{{\cal M}}

\newcommand{\mm}{\MM'}

\let \PROC \algorithmicproc
\definecolor{gray}{RGB}{72,61,139}

\renewenvironment{algorithmic}{\algsetup{indent=2.5em,linenosize=\small,linenodelimiter=.}\begin{oldalgorithmic}}{\end{oldalgorithmic}}

\definecolor{blue}{rgb}{0.274,0.392,0.666}
\definecolor{red}{rgb}{0.827,0.017,0.056}
\definecolor{green}{rgb}{0,0.588,0.509}

\begin{document}

\title{\MakeUppercase{How to Morph Planar Graph Drawings}
\footnote{amalgamates results from conference papers~\cite{SODA-morph,Barrera-unidirectional,angelini2013morphing,Angelini-optimal-14}.}
}

\author{
Soroush Alamdari%
	\thanks{\affil{Department of Computer Science, Cornell University, Ithaca, New York, USA}, 
          \email{alamdari@cs.cornell.edu}},
Patrizio Angelini%
\thanks{\affil{Wilhelm-Schickard-Institut f\"{u}r Informatik, T\"ubingen University, Germany}, 
	\email{angelini@informatik.uni-tuebingen.de}},
Fidel~Barrera-Cruz%
\thanks{\affil{School of Mathematics, Georgia Institute of Technology, Atlanta, Georgia, USA}, 
          \email{fidelbc@math.gatech.edu}},
Timothy~M.~Chan%
\thanks{\affil{David R. Cheriton School of Computer Science, University of Waterloo, Waterloo, Canada}, 
          \email{\{tmchan, alubiw\}@uwaterloo.ca}},
Giordano Da Lozzo%
\thanks{\affil{Department of Engineering, Roma Tre University, Italy}, 
	\email{\{dalozzo, gdb, frati, patrigna, roselli\}@dia.uniroma3.it}},
Giuseppe~Di~Battista\footnotemark[5],
Fabrizio~Frati\footnotemark[5],
Penny~Haxell%
\thanks{\affil{Department of Combinatorics and Optimization, University of Waterloo, Waterloo, Canada}, 
          \email{pehaxell@uwaterloo.ca}},
Anna~Lubiw\footnotemark[4],
Maurizio~Patrignani\footnotemark[5],
Vincenzo~Roselli\footnotemark[5],
Sahil~Singla%
\thanks{\affil{School of Computer Science, Carnegie Mellon University, Pittsburgh, Pennsylvania, USA}, 
          \email{ssingla@cmu.edu}},
and Bryan~T.~Wilkinson%
\thanks{\affil{Center for Massive Data Algorithmics, Aarhus University, Aarhus, Denmark}, 
          \email{wilkinson@cs.au.dk}}
}

\maketitle

\begin{abstract}	
Given an $n$-vertex graph and two straight-line planar drawings of the graph that have the same faces and the same outer face, we show that there is a morph (i.e., a continuous transformation) between the two drawings that preserves straight-line planarity and consists of $O(n)$ steps, which we prove is optimal in the worst case.
Each step is a \emph{unidirectional linear morph}, which means that	every vertex moves at constant speed along a straight line, and the lines are parallel although the vertex speeds may differ. Thus we provide an efficient version of Cairns' 1944 proof of the existence of straight-line planarity-preserving morphs for triangulated graphs, which required an exponential number of steps.
\end{abstract}

\section{Introduction}\label{se:intro}

A morph between two geometric shapes is a continuous transformation of one shape into
the other. 
Morphs are useful in many areas of computer science---Computer Graphics, Animation, and Modeling, to name just a few.
The usual goal in morphing is to ensure that the structure of the shapes be ``visible'' throughout the entire transformation.


Two-dimensional graph drawings can be used to represent many of the shapes for which morphs are of interest, e.g., two-dimensional
images~\cite{bn-fbim-92,fm-ffip-98,te-imfpt-99}, polygons, and  poly-lines~\cite{aadddhlrssw-cpwlv-11,acl-arpsi-00,cc-wgosm-97,elrsw-cssp-98,gh-msp-94,nmwb-mpstcg-08,sgwm-2dsbisvpp-93,sg-pba2dsb-92,sr-sbussr-95}. 
For this reason, morphs of graph drawings have been well studied. 

For the problem of morphing graph drawings,  
the input consists of two 
drawings $\Gamma_0$ and $\Gamma_1$ of the same graph $G$, and the problem is to transform continuously from the first drawing to the second drawing.
A \emph{morph} between $\Gamma_0$ and $\Gamma_1$ is  a continuously changing family of drawings of $G$ indexed by time $t \in  [0, 1]$, such that the drawing at time $t = 0$ is $\Gamma_0$ and the drawing at time $t = 1$ is $\Gamma_1$. 
Maintaining structure during the morph becomes a matter of preserving geometric properties such as planarity, straight-line planarity, edge lengths, or edge directions.
For example, preserving edge lengths in a straight-line drawing leads to problems of linkage reconfiguration~\cite{CDR,DO'R}.

In this paper we consider the problem of morphing between two graph drawings while preserving planarity.
Of necessity, we 
assume that the initial and final planar drawings are \emph{topologically equivalent}---i.e., have the same faces and the same outer face.
In addition to the above-mentioned applications, morphing graph drawings while preserving planarity has application to the problem of creating three-dimensional models from two-dimensional slices~\cite{BS}, with time playing the role of the third dimension.

When planar graphs may be drawn with poly-line edges the morphing problem becomes much easier---the intuition is that vertices can move around while edges bend to avoid collisions.  An efficient morphing algorithm for this case was given by 
Lubiw and Petrick~\cite{LP}. 
The case of orthogonal graph drawings is also well-solved---Biedl et al.~\cite{biedl2013morphing} gave
an algorithm to morph efficiently between any two orthogonal drawings of the same graph while preserving planarity and orthogonality.

We restrict our attention in this paper to straight-line planar drawings.
Our main result is an efficient algorithm to morph between two topologically equivalent straight-line planar drawings of a graph, where the morph must preserve straight-line planarity.  The issue is to find the vertex trajectories, since the edges are determined by the vertex positions.

Existence of such a morph is not obvious, and was first proved in 
1944 by Cairns~\cite{c-dplc-44} for the case of triangulations.
Cairns used an inductive proof, based on contracting a low-degree vertex to a neighbor.
In general,  a contraction that preserves planarity in both drawings may not exist, so Cairns needed a preliminary morphing procedure to make this possible.  As a result, his method 
involved two recursive calls, and took an exponential number of steps.
Thomassen~\cite{t-dpg-83} extended the proof to all planar straight-line drawings. He did this by augmenting both drawings to  isomorphic (``compatible'') triangulations which reduces the general case to Cairns's result. The idea of compatible triangulations was rediscovered and thoroughly explored by Aronov et al.~\cite{ASS}, who showed, among other things, that two drawings of a graph on $n$ vertices have a compatible triangulation of size $O(n^2)$ and that this bound is tight in the worst case.
 
Floater and Gotsman~\cite{FG} gave an alternative way to morph straight-line planar triangulations based on Tutte's graph drawing algorithm~\cite{Tutte}.
Gotsman and Surazhsky~\cite{GS,surazhsky-mpt-99,SG,SG2} extended the method to all straight-line planar graph drawings using the same idea of compatible triangulations, and they showed that the resulting morphs are visually appealing.
These algorithms do not produce explicit vertex trajectories; instead, they compute the intermediate drawing (a ``snapshot''') at any requested time point.  There are no quality guarantees about the number of time-points required to approximate continuous motion while preserving planarity.  
For related results, see~\cite{ekp-ifmpg-03,ekp-mpg-04,fe-gdm-02}.

For more history and related results on morphing graph drawings, see Roselli's PhD thesis~\cite[Section 3.1]{Roselli-thesis}.

The problem of finding a straight-line planar morph that uses a polynomial number of discrete steps has been asked several times
 (see, e.g.,~\cite{km-mpgss-08,l-mpgd-07,LP,LPS}).  
The most natural definition of a discrete step is a \emph{linear morph}, where every vertex moves along a straight-line segment at uniform speed.  
Note that we do not require that all vertices move at the same speed.
One of the surprising things we discovered (after our first conference version~\cite{SODA-morph}) is that it is actually easier to solve our problem using a more restrictive type of linear morph.  
Specifically, we define a morph to be \emph{unidirectional} if every vertex moves along a straight-line segment at uniform speed, and all the lines of motion are parallel. As a special case, a linear morph that only moves one vertex is unidirectional by default.  

\subsection{Main Result}
\label{sec:main-result}

The main result of this paper is the following:

\begin{theorem}\label{th:main}
Given a planar graph $G$ on $n$ vertices and two straight-line planar drawings of $G$ with the same faces, the same outer face, and the same nesting of connected components, there is  
a morph between the two drawings that preserves straight-line planarity and consists of $O(n)$ unidirectional morphs. Furthermore, the morph can be found in time $O(n^3)$. 
\end{theorem}

This paper combines four conference papers: \cite{SODA-morph} designs a general algorithmic scheme for constructing morphs between planar graph drawings and proves the first polynomial upper bound on the number of morphing steps; \cite{Barrera-unidirectional}~introduces unidirectional morphs; \cite{angelini2013morphing}~introduces techniques to handle non-triangulated graphs; \cite{Angelini-optimal-14} solves a crucial subproblem of   ``convexifying a quadrilateral'' with a single unidirectional morph, yielding a linear bound on the total number of morphing steps (which they prove optimal).  New to this version is the handling of disconnected graphs.


Techniques from our paper have been used in algorithms to morph Schnyder drawings~\cite{Barrera-Cruz-Schnyder}, and algorithms to morph convex drawings while preserving convexity~\cite{adflpr-omcd-15}. 

\remove{
We note that a newer result~\cite{Angelini-optimal-14} reduces the bound in Theorem~\ref{th:main} to a total of $O(n)$ unidirectional morphs, which is optimal. While this new result supersedes ours, all its main ideas are borrowed from our paper, with the exception of an elegant solution to the Quadrilateral Convexification problem with $O(1)$ unidirectional morphs, achieved via a previously unseen connection between that problem and the existence of hierarchical planar convex drawings of hierarchical planar triconnected graphs. Beside the use of this connection, the algorithm in~\cite{Angelini-optimal-14} is {\em exactly} the same as ours, which justifies our belief that this paper has been foundational to all subsequent research done in this topic (see~\cite{Angelini-optimal-14},~\cite{Barrera-Cruz-Schnyder} for algorithms to morph Schnyder drawings, and~\cite{adflpr-omcd-15} for algorithms to construct convex morphs).
}

From a high-level perspective, 
our proof of Theorem~\ref{th:main} has two parts: to solve the problem for the special case of a maximal planar graph, in which case both drawings are triangulations; and to reduce the general problem to this special case. 

Previous papers~\cite{SODA-morph,GS,t-dpg-83} reduced the general case to the case of triangulations by finding 
``compatible'' triangulations of both drawings, which increases the size of the graph to $O(n^2)$.
We improve this by making use of the freedom to morph the drawings.  
Specifically, we show that after a sequence of $O(n)$ unidirectional morphs 
we can triangulate both drawings with the same edges.
Thus we reduce the general problem to the case of triangulations with the same input size.

For the case of triangulations, the main idea of our algorithm is the same one that Cairns used to prove existence of a morph, namely to find a vertex $v$ that can be contracted in both drawings to a neighbor $u$ while preserving planarity.
Contracting $v$ to $u$ gives us two planar drawings of a smaller graph.
Using recursion, we can find a morph between the smaller drawings that consists of unidirectional morphs.
Thus the total process is to move $v$ along a straight line to $u$ in the first drawing (a unidirectional morph since only one vertex moves), perform the recursively computed morph, and then, in the second drawing, reverse the motion of $v$ to $u$.  

Note that this process allows a vertex to become coincident with another vertex, so it is not a true morph, but rather what we call a  ``pseudomorph''. 
There are two main issues with this plan: 
 (1) to deal with the fact that  there may not be a vertex $v$ that can be contracted to the same neighbor $u$ in both drawings; and 
(2) to convert a pseudomorph to a morph. 
We will give a few more details on each of these.

A contractible vertex always exists if we are dealing with a single triangulation:
because the graph is planar, there is an internal vertex  $v$ of degree less than or equal to 5; because the graph is triangulated, the neighbors of $v$ form a polygon of at most 5 vertices; and by an easy geometric argument, such a polygon always has a vertex $u$ that ``sees'' the whole polygon, so $v$ can be contracted to $u$ while preserving planarity.  In our situation there is one complication---we want to contract $v$ to the same neighbor in both drawings.
This is easy to solve in the case where our low-degree vertex $v$ has neighbors that form a convex polygon in the first drawing---in this case every neighbor of $v$ sees the whole convex polygon in the first drawing, so we can choose the vertex $u$ that works in the second drawing.  Our general solution will be to morph the first drawing so that $v$'s neighbors form a convex polygon (or at least ``convex enough'', in the sense that $u$ sees the whole polygon).   This was the same approach that Cairns used, though his solution required an exponential number of morphing steps, and our solution will only require two morphing steps.

To summarize the first part of our algorithm to morph between two triangulations: we show that after two unidirectional morphs we can obtain a vertex $v$ that can be contracted to the same neighbor in both drawings. 

After performing the contraction we apply induction to find a morph (composed of unidirectional morphs) between the two smaller drawings.  
The last issue is to convert the resulting pseudomorph to a true morph.
Instead of contracting $v$ to a neighbor $u$, we must keep $v$ close to, but not coincident with, $u$, while we follow the morph of the smaller graph. 
Cairns solved this issue by keeping $v$ at the centroid of its surrounding polygon, but this results in 
non-linear motion for $v$~\cite{SODA-morph}.
We will find a position for $v$ in each drawing during the course of the morph so that the linear motion from one drawing to the next remains unidirectional.

Putting together the two parts of our algorithm, the total number of unidirectional morphs satisfies $S(n) = S(n-1) + O(1)$, which solves to $O(n)$.  

This completes the high-level overview of our algorithm.  From a low-level perspective, the 
heart of our algorithm is a solution to a problem we call \emph{Quadrilateral Convexification}: given a triangulation containing a non-convex quadrilateral, morph the triangulation to make the quadrilateral convex.

Our morphing algorithm uses $O(n)$ calls to Quadrilateral Convexification, plus $O(n)$ other unidirectional morphs.  We will show that Quadrilateral Convexification can be accomplished with a single unidirectional morph, and thus our total bound is $O(n)$ unidirectional morphs. Our solution to the Quadrilateral Convexification problem is achieved via a connection to the existence of hierarchical planar convex drawings of hierarchical triconnected planar graphs.

We also show that the linear bound of Theorem~\ref{th:main} is asymptotically optimal in the worst case. Namely, we have the following. 

\begin{theorem}\label{th:lb-bound}
There exist two straight-line planar drawings of an $n$-vertex path such that any planar morph between them consists of $\Omega(n)$ linear morphs.
\end{theorem}

In particular, we show that morphing from an $n$-vertex spiral to a straight path takes $\Omega(n)$ linear morphs by defining a measure of the difference between the two drawings that begins at $\Omega(n)$ and changes only by $O(1)$ during a single linear morph.

\remove{
In order to prove Theorem~\ref{th:lb-bound}, we exploit one drawing in which all the edges lie on a common straight line and another drawing in which the edges are ``wrapped'' in a spyral-like fashion. We define a measure representing, given two drawings of a path, how differently the two drawings are wrapped; then we prove that the measure for the two considered drawings is $\Omega(n)$ and that a single linear morph only changes the measure by $O(1)$. This leads to prove Theorem~\ref{th:lb-bound}.}

The rest of the paper is organized as follows. First, in
Section~\ref{sec:definitions} we give formal definitions of our terms and concepts.
Then in Section~\ref{se:overview} we give a more detailed outline of the algorithm.  We fill in solutions to the various subproblems in Sections~\ref{sec:triangulate}--\ref{sec:geometry}. In Section~\ref{se:lowerbound} we present our lower bound. Finally, in Section~\ref{se:conclusions} we conclude and present some open problems.

\section{Definitions and Preliminaries}
\label{sec:definitions}

A \emph{straight-line planar drawing} $\Gamma$ of a graph $G(V,E)$ maps vertices in $V$ to distinct points of the plane and edges in $E$ to non-intersecting open straight-line segments between their end-vertices.
A planar drawing of a graph partitions the plane into topological connected regions called \emph{faces}. The unbounded face is called the \emph{outer face}. 
Two planar drawings of a connected planar graph are \emph{topologically equivalent} if they induce the same circular ordering of the edges around each vertex and have the same outer face.   
Two planar drawings of a disconnected planar graph are \emph{topologically equivalent} if each connected component is topologically equivalent in both drawings, and furthermore, the connected components are nested the same way in both drawings.
A \emph{planar embedding} is an equivalence class of planar drawings of the same graph.
A \emph{plane} graph is a planar graph with a given planar embedding.

Given a vertex $v$ of a graph $G$, the \emph{neighbors} of $v$ are the vertices adjacent to $v$, and the  \emph{degree} of $v$ in $G$, denoted by $\deg(v)$, is the number of neighbors of $v$.

In a plane graph, a \emph{facial cycle} is a closed walk that progresses from one edge $xy$ to the next edge $yz$ in the clockwise cyclic order of edges around vertex $y$.  
In a planar drawing, each inner face is bounded by an outer facial cycle and some number of inner facial cycles, one for each connected component that is drawn inside the face.

\smallskip
\noindent{\bf Morphs.}
If $\Gamma_0$ and $\Gamma_1$ are two 
drawings of the same graph, a \emph{morph} between $\Gamma_0$ and $\Gamma_1$ is  a continuously changing family of drawings of $G$ indexed by time $t \in  [0, 1]$, such that the drawing at time $t = 0$ is $\Gamma_0$ and the drawing at time $t = 1$ is $\Gamma_1$. In this paper we are only concerned with graph drawings in which every edge is drawn as a straight-line segment.  In this case, a morph is specified by the vertex trajectories. 

A \emph{linear morph} is a morph in which every vertex moves along a straight-line segment at uniform speed.
A linear morph is completely specified by the initial and final vertex positions.
If vertex $v$ is at position $v_0$ in the initial drawing (at time $t=0$) and at position $v_1$ in the final drawing (at time $t=1$), then its position at time $t$ during a linear morph is $(1-t)v_0 + tv_1$, for any $0\leq t\leq 1$.  Note that vertices may move at different speeds, and in particular, some vertices may remain stationary.  

If $\Gamma_0$ and $\Gamma_1$ are straight-line planar drawings of a graph, we use $\morph{\Gamma_0, \Gamma_1}$ to denote the linear morph from $\Gamma_0$ to $\Gamma_1$. 
We seek a morph that consists of a sequence of $k$ linear morphs.  
Such a morph can be specified by $k+1$ planar straight-line graph drawings.   If $\Gamma_1, \ldots, \Gamma_{k+1}$ are straight-line planar drawings of a graph, we use $\morph{\Gamma_1, \ldots, \Gamma_{k+1}}$ to denote the morph from $\Gamma_1$ to $\Gamma_{k+1}$ that consists of the sequence of $k$ linear morphs $\morph{\Gamma_i,\Gamma_{i+1}}$, for $i=1, \ldots, k$.

A \emph{unidirectional morph} is a linear morph in which every vertex moves parallel to the same line, i.e.~there is a line $L$ with unit direction vector $\bar \ell$ such that each vertex moves linearly from an initial position $v_0$ to a final position $v_0 + k_v {\bar \ell}$ for some $k_v \in {\mathbb R}$. 
Note that $k_v$ may be positive or negative and that different vertices may move different amounts along direction $\bar \ell$.  We call this an \emph{$L$-directional morph}.
Observe that a linear morph of a single vertex is by default a unidirectional morph.

In this paper we restrict attention to 
topologically equivalent
straight-line planar drawings and to morphs in which every intermediate drawing is straight-line planar.  
From now on we use the term \emph{morph} to mean a \emph{straight-line planarity preserving morph}. 
In particular, during the course of the morph, a vertex may not become coincident with another vertex, nor hit a non-incident edge.

In several places we will make use of the following basic result: If we  have a straight-line planar-drawing 
of a maximal planar graph, 
and if we apply a unidirectional morph to the three vertices of the triangular outer face, preserving the orientation of the triangle,  and let the interior vertices follow along linearly, then the result is a unidirectional morph.

\begin{lemma} Let $x,y,z$ be the clockwise vertices of the triangular outer face of a straight-line planar drawing of a maximal planar graph.  Suppose that vertices $x$, $y$, and $z$ move linearly in the direction of a vector $\bar{\ell}$ in such a way that their clockwise order is preserved.  Any point $p$ inside the triangle can be defined as a convex combination of $x$, $y$, $z$, and in this way the motion of $x$, $y$, $z$ determines the motion of  $p$.  The result is a unidirectional morph of the straight-line planar drawing (in particular, planarity is preserved). 
  \label{lemma:convex-comb}
\end{lemma}
\begin{proof}
Suppose that point $p$ is defined by the convex combination 
$\lambda_1 x + \lambda_2 y + \lambda_3 z$ where $\sum \lambda_i = 1$
and $\lambda_i \ge 0$. 
  Suppose the morph is indexed by $t \in [0,1]$ and that the positions
  of the vertices at time $t$ are $x_t, y_t, z_t, p_t$.  Suppose that
  $x$ moves by $k_1 \bar \ell$, $y$ moves by $k_2 \bar \ell$, and $z$
  moves by $k_3 \bar \ell$.  Thus $x_t = x_0 + tk_1 \bar{\ell}$ etc.
  Then $$p_t = \lambda_1 x_t + \lambda_2 y_t + \lambda_3 z_t \\ =
  \lambda_1 x_0 + \lambda_2 y_0 + \lambda_3 z_0 + t(\lambda_1 k_1 +
  \lambda_2 k_2 + \lambda_3 k_3) \bar{\ell} \\= p_0 + tk \bar{\ell}$$ where
  $k = \lambda_1 k_1 + \lambda_2 k_2 + \lambda_3 k_3$.  Thus $p$ also
  moves linearly in direction $\bar \ell$.
  
The fact that planarity is preserved follows from far more general results: The transformation of points $x$, $y$, $z$ determines an affine transformation of the plane, and by hypothesis, this affine transformation preserves the orientation of triangle $xyz$.   Affine transformations preserve convex combinations---thus our definition of the movement of any interior point $p$ is the same as applying the affine transformation to $p$.  An affine transformation that preserves the orientation of one triangle preserves the orientations of all triangles.  This implies that the drawing is planar at all time points of the morph. 
\end{proof}

\remove{ 
A \emph{(planar linear) morphing step} \mmorph{\Gamma_1,\Gamma_2}, also referred to as
\emph{linear morph} or \emph{step}, of two straight-line planar drawings $\Gamma_1$ and
$\Gamma_2$ of a plane graph $G$ is a continuous transformation of $\Gamma_1$
into $\Gamma_2$ such that: 
\begin{itemize}
\item all the vertices simultaneously start moving from their positions in $\Gamma_1$; 
\item each vertex moves at constant speed along a straight-line trajectory; 
\item all the vertices simultaneously stop at their positions in $\Gamma_2$;
\item no crossing occurs between any two edges during the transformation; and
\item no two vertices are mapped to the same point during the transformation.
\end{itemize}

A \emph{morph} \mmorph{\Gamma_s,\dots,\Gamma_t} of two straight-line planar
drawings $\Gamma_s$ into $\Gamma_t$ of a plane graph $G$ is a finite sequence of
morphing steps that transforms $\Gamma_s$ into $\Gamma_t$.
}

\smallskip
\noindent{\bf Geometry and Triangulations.}
We will assume that our input graph drawings have vertices in general position, that is, no three vertices lie on the same line.  
We can achieve this property by a linear number of preliminary unidirectional morphing steps that slightly perturb the positions of the vertices.

A \emph{triangulation} is a straight-line planar drawing of a maximal planar graph.  Every face in a triangulation (including the outer face) is a triangle.  The three vertices of the outer face are called \emph{boundary vertices} and the others are called \emph{internal} or \emph{non-boundary} vertices.

If $v$ is an internal vertex of a triangulation, we use $\Delta(v)$ to denote the polygon formed by the neighbors of $v$.
For a simple polygon in the plane, the \emph{kernel} of the polygon consists of the points inside the polygon  from which the whole polygon is visible.  Note that the kernel of any polygon is convex.
The following result
was noted by Cairns and can be proved by simple case analysis.

\begin{lemma}
If $P$ is a polygon with four or five vertices (a quadrilateral or a pentagon) then at least one vertex of $P$ is contained in the kernel of $P$.
\label{lem:5-kernel}
\end{lemma}

We will use the result in the following form: If $v$ is an internal vertex of degree at most $5$ in a triangulation, then $\Delta(v)$ has a vertex in its kernel.

\smallskip
\noindent{\bf Contractions and Pseudomorphs.}
A main tool we use in our morphing algorithm is vertex contraction in a triangulation. 
Contracting edge $uv$ in a graph has the standard meaning, namely we replace $u$ and $v$ by a new vertex adjacent to all the neighbors of $u$ and $v$. 
We now define contraction in a triangulation.
Let $\Gamma$ be a drawing of a maximal planar graph $G$.  Let $v$ be an internal vertex and let 
$u$ be a neighbor of $v$ that lies in the kernel of $\Delta(v)$.  \emph{Contracting} $v$ to $u$ means moving $v$ linearly from its original position to $u$ while all other vertices remain fixed.  Because the kernel is convex, every intermediate drawing is straight-line planar.  Thus, this is a morph (and in fact a unidirectional morph) except for the final drawing in which $v$ becomes coincident with $u$.  By our general position assumption the final drawing is a straight-line planar drawing of the graph formed by contracting edge $uv$.


\remove{ 
As in the works of Cairns~\cite{c-dplc-44} and Thomassen~\cite{t-dpg-83}, our approach makes an extensive use of contractions, uncontractions, and recursion. Namely, we aim at
recursively computing a morph of the input drawings by contracting, at each step, a vertex $v$ onto the same neighbor $x$ in both the source and the target drawings $\Gamma_s$ and $\Gamma_t$, respectively. Since a vertex $v$ that can be contracted onto the same neighbor $x$ in $\Gamma_s$ and $\Gamma_t$ might not exist, we may need to suitably morph $\Gamma_s$ and $\Gamma_t$ into two drawings $\Gamma_s^x$ and $\Gamma_t^x$, respectively, in which such a contraction is possible. In order to make these transformations possible, as in~\cite{c-dplc-44} and in~\cite{t-dpg-83}, we focus on low-degree ``candidate'' vertices, leveraging on their topological and geometric properties.

Namely, we say that a vertex $v$ is a \emph{\cand to the contraction} (or simply a \emph{\cv}) if: 
\begin{enumerate}[$(i)$]
\item $\deg(v)\le 5$; and
\item if two of its neighbors $u$ and $w$ are connected by an edge, then \face{u,v,w} is a simple face of $G$. 
\rednote{AL. ``Simple'' does not quite capture what we want here.}
\end{enumerate}

In Section~\ref{se:candidate_existence} we show that it is always possible to find a \cv in a plane graph while, in Section~\ref{se:xcontractible}, we describe how to compute a morph that makes \xc a \cv.
}


Our algorithm for morphing between two triangulations $\Gamma_1$ and $\Gamma_2$ works by contracting some vertex $v$  to the same neighbor $u$ in both drawings, and then recursively morphing between the two smaller triangulations.  
Expressing this as a transformation from $\Gamma_1$ to $\Gamma_2$, we contract $v$ to $u$ in $\Gamma_1$, apply the recursively computed morph, and then reverse the contraction of $v$ to $u$ in $\Gamma_2$.  We call this last step ``uncontraction'' of $v$.
The complete transformation is called a \emph{pseudomorph}.  It is not a morph    
because $v$ becomes coincident with $u$. 
One of our main technical contributions is to show that every pseudomorph that is composed of unidirectional morphs can be converted to a ``true'' morph with the same number of unidirectional steps.

Our model of computing is a real RAM. 

\remove{
Let $\Gamma_1$ and $\Gamma_2$ be two drawings of the same plane graph $G$ in which $v$ is contractible onto the same neighbor $x$. 
We define a \emph{pseudo-morph} of $\Gamma_1$ into $\Gamma_2$ as a sequence of operations composed of:
\rednote{AL.  Define a pseudomorph to be a sequence of contraction, uncontraction, morph steps (as in the original SODA version).  Although we will only use one <contraction, morph, uncontraction> in the main loop, it is useful to have the richer set for the 4-gon and 5-gon convexification.  Define the number of steps of a pseudomorph to count contraction and uncontraction as steps.} 
\begin{enumerate}[$(i)$]
\item the contraction of $v$ onto $x$ in $\Gamma_1$, resulting in a drawing $\Gamma_1'$ of $G'=G/(v,x)$;
\item the morph of $\Gamma_1'$ into a drawing $\Gamma_2'$ of $G'$; and
\item the uncontraction of $v$ in $\Gamma_2'$ resulting in $\Gamma_2$.
\end{enumerate}

Section~\ref{se:geometry} describes how to compute a morph from $\Gamma_1$ to $\Gamma_2$ starting from a pseudo-morph from $\Gamma_1$ to $\Gamma_2$.
}

\section{Overview of the Algorithm}
\label{se:overview}

In this section we describe all the ingredients of our algorithm and how they fit together.

Our most basic building block is an algorithm to morph a triangulation so that a given quadrilateral formed by two adjacent triangles becomes convex.  One chord of the quadrilateral will lie inside the quadrilateral.  A necessary condition is that the other chord not be part of the triangulation. Specifically, we solve the following problem:

\begin{problem}[Quadrilateral Convexification]
Given an $n$-vertex triangulation $\Gamma$ and given a quadrilateral $abcd$ in $\Gamma$ with no vertex inside it  and such that neither $ac$ nor $bd$ is an edge outside of $abcd$ (i.e.,~$abcd$ does not have external chords), 
morph $\Gamma$ so that $abcd$ becomes convex.
\label{prob:quad-convex}
\end{problem}

We solve Quadrilateral Convexification in Section~\ref{sec:quad-convexification} giving the following result: 

\def\quadconvex{Quadrilateral Convexification can be solved via a 
single unidirectional morph.
Furthermore, 
such a morph can be found in $O(n^2)$ time.}
\begin{theorem}
\quadconvex
\label{thm:quad-convex}
\end{theorem}

We will prove our main result, Theorem~\ref{th:main}, by finding a morph that consists of 
$O(n)$ calls to Quadrilateral Convexification plus $O(n)$ further unidirectional morphs.
Together with the above Theorem~\ref{thm:quad-convex}, this gives a total bound of $O(n)$ unidirectional morphs.

%
%

The proof of Theorem~\ref{th:main} has two parts.
In Section~\ref{sec:triangulate} we reduce the problem to the case of triangulations.  Specifically, we show that given two topologically equivalent straight-line planar drawings of a graph $G$ on $n$ vertices, 
we can enclose each drawing with a triangle and then morph and add (the same) edges to triangulate both drawings, using $O(n)$ unidirectional morphs; this results in two triangulations which are topologically equivalent drawings of the same maximal planar graph. Note that a morph between these augmented drawings provides a morph of the originals by ignoring the added edges.
  
In Section~\ref{sec:morph-triangulations} we prove Theorem~\ref{th:main} for the case of triangulations, using $O(n)$ unidirectional morphs. Thus the two sections together prove Theorem~\ref{th:main}.

In all cases when we say that we find a morph, we actually find a pseudomorph (as defined in the previous section) and rely on the following theorem (which is proved in Section~\ref{sec:geometry}) to convert the pseudomorph to a true morph:

\def\pseudomorph{Let $\Gamma_1$ and $\Gamma_2$ be two triangulations that are topologically equivalent drawings of an $n$-vertex maximal planar graph $G$.
Suppose that  there is a pseudomorph from $\Gamma_1$ to $\Gamma_2$ 
in which we contract an internal vertex $v$ of degree at most 5, perform $k$ unidirectional morphs, and then uncontract $v$.
Then there is a morph $\cal M$ from $\Gamma_1$ to $\Gamma_2$ that consists of $k+2$ unidirectional morphs.  Furthermore, given the sequence of $k+1$ drawings that define the pseudomorph, we can modify them to obtain the sequence of drawings that define $\cal M$ in $O(k+n)$ time.}
\begin{theorem}
\pseudomorph
\label{thm:pseudomorph}
\end{theorem}

Note that although the proofs of Theorems~\ref{thm:quad-convex} and~\ref{thm:pseudomorph} are deferred until the last two sections of the paper, they come first in terms of the dependency of results.

\remove{
\begin{algorithm}[tb]
\renewcommand{\thealgorithm}{\texttt{Morph}} 
\caption{\PROC{}{$\Gamma_s$,$\Gamma_t$}}\label{algo:morph}
\begin{algorithmic}[1]
\algsetup{indent=1em,linenosize=\small,linenodelimiter=.}
\medskip
\REQUIRE{$\Gamma_s$ and $\Gamma_t$ are two straight-line planar drawings of a connected plane graph $G$}\\
\medskip
\COMMENT{preprocessing}
\STATE{morph and augment so that both drawings are triangulations with fixed outer triangle (Section~\ref{sec:triangulate})}\\
\COMMENT{base case}
\IF{$G$ is a triangle}
\STATE{return the identity morph}\label{line:morph2}
\ENDIF\vspace{.3em}
\COMMENT{recursive step}
\STATE $v \gets$ an internal vertex of degree $\le 5$\label{line:morph4}
\STATE $u \gets$ a neighbor of $v$ to which $v$ can be contracted in $\Gamma_t$\label{line:morph5}
\STATE morph $\Gamma_s$ so that $v$ can be contracted to $u$ \label{line:morph6} 
\STATE $\Gamma_s^u \gets$ contract $v$ to $u$ in $\Gamma_1$\label{line:morph7}
\STATE $\Gamma_t^u \gets$ contract $v$ to $u$ in $\Gamma_2$\label{line:morph8}
\STATE ${\cal M} \gets$ {\sc Morph}$(\Gamma_s',\Gamma_t')$\label{line:morph9}\vspace{.3em}\\
\RETURN PseudoMorph-to-Morph(<$\Gamma_1, \Gamma_1^u$> $ {\cal M}$ <$\Gamma_2^u, \Gamma_2$ >)
\end{algorithmic}
\end{algorithm}
}

\section{Morphing to Find a Compatible Triangulation}
\label{sec:triangulate}

In this section we show how to morph two topologically equivalent straight-line planar drawings of a graph $G$ so that after the morph both drawings can be triangulated by adding the same edges.
We allow $G$ to be disconnected.

\begin{theorem}
\label{thm:compatible-triangulation}
Let $G$ be a planar graph with $n$ vertices and $c$ connected components.
Given two topologically equivalent straight-line planar drawings of $G$, we can enclose each drawing in a triangle $z_1 z_2 z_3$ and then morph and add edges to create two  
triangulations that are topologically equivalent drawings of a maximal planar graph on vertex set $V(G) \cup \{z_1,z_2,z_3\}$. The morph consists of $O(c)$ unidirectional morphs plus $O(n)$ calls to Quadrilateral Convexification, for a total bound of $O(n)$ unidirectional morphs.
The algorithm takes time $O(n^3)$.
\end{theorem}

\begin{proof}
Let $\Gamma_1$ and $\Gamma_2$ be the two topologically equivalent drawings of $G$.  We begin by adding a large triangle $z_1 z_2 z_3$ that encloses each drawing.  

Our algorithm has two parts.  In part A we morph and add edges within connected components so that in any connected component of three or more vertices, each facial walk is a triangle.
In part B we add edges to connect the disconnected components and complete the triangulation. 

\medskip\noindent
{\bf Part A.}
Suppose we have a facial walk that has four or more vertices.  We will find two consecutive edges
 $uv$ and $vw$ of the facial walk so that edge $uw$ can be added to the graph, i.e.~such that $u\neq w$ and $uw$ is not an edge of the graph. Suppose first that the facial walk has two consecutive edges $x_1 x_2$ and $x_2 x_3$ that belong to different biconnected components of the graph, i.e., the removal of $x_2$ disconnects the connected component $x_2$ belongs to. Then $x_1\neq x_3$ and edge $x_1 x_3$ does not belong to the graph, hence it can be added inside the face. If two such edges $x_1 x_2$ and $x_2 x_3$ do not exist, then the facial walk is a simple cycle. Then consider four consecutive vertices $x_1$, $x_2$, $x_3$, and $x_4$ along the cycle. By planarity, edge $x_1x_3$ or edge $x_2 x_4$ does not belong to the graph, hence it can be added inside the face.

 
At this point we have two consecutive edges $uv$ and $vw$ of the facial walk so that edge $uw$ can be added to the graph while preserving the planarity of the graph, as in Figure~\ref{fig:v_noedge}. We will morph drawings $\Gamma_1$ and $\Gamma_2$ so that edge $uw$ can be added as a straight-line segment preserving planarity. 
 (In case our facial walk bounds an inner face that contains no disconnected components, we could have chosen $uw$ to be an ear of a triangulation of the face in one of the drawings, thus avoiding the need to morph that drawing; but in the general case we must morph both drawings.) 
The argument is the same for both drawings, so we describe it only for $\Gamma_1$.  
Vertices $u$ and $w$ are consecutive neighbors of $v$.  Add a new neighbor $r$ of $v$ between $u$ and $w$ in cyclic order, and add  edges $rv$, $ru$, and $rw$.  It is possible to place $r$ close enough to $v$ in $\Gamma_1$ so that the resulting drawing is straight-line planar and  \face{v,r,u} and \face{v,r,w} are faces. 
See Figure~\ref{fig:v_dummy} for an example.

\begin{figure}[htb]
\centering
\def\sf{1.2}

{\hfill
\subfigure[]{\includegraphics[scale=\sf]{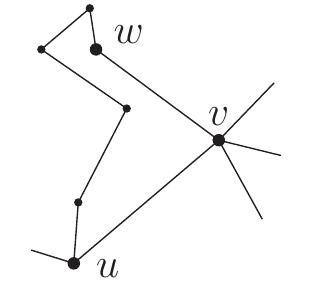}\label{fig:v_noedge}}\hfill
\subfigure[]{\includegraphics[scale=\sf]{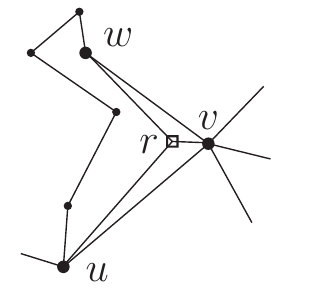}\label{fig:v_dummy}}\hfill }

{\hfill
\subfigure[]{\includegraphics[scale=\sf]{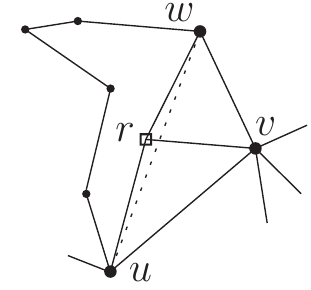}\label{fig:v_convex}}\hfill
\subfigure[]{\includegraphics[scale=\sf]{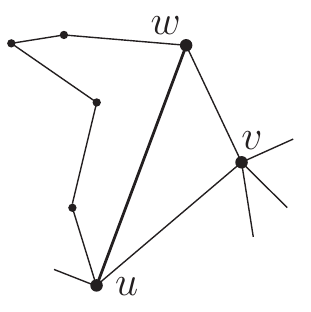}\label{fig:v_edge}}\hfill}
\caption{Morphing to add edge $uw$.   
\subref{fig:v_noedge}~Vertices $u, v$ and $w$ are consecutive around a facial walk, and the graph does not currently contain the edge $uw$. 
\subref{fig:v_dummy}~A vertex $r$ is added suitably close to $v$ and connected to $v$, $u$, and $w$. \subref{fig:v_convex}~After convexifying the quadrilateral $u r w v$. 
\subref{fig:v_edge}~Vertex $r$ and its incident edges can be removed in order to insert edge $uw$.\label{fig:conv_face}}
\end{figure}

We will now morph the resulting drawing 
to make the quadrilateral $urwv$ convex, as in Figure~\ref{fig:v_convex}.  To do this, we temporarily triangulate the drawing\footnote{Any polygonal domain with $n$ vertices can be triangulated in time $O(n \log n)$~\cite{deBerg}.}  and then apply Quadrilateral Convexification to $urwv$.

There is one slight complication: it might happen that when we triangulate the drawing, we add the edge $uw$, which would make it impossible to convexify the quadrilateral $urwv$.  In this case, we remove $uw$ and retriangulate the resulting quadrilateral by adding a new vertex $p$ placed at any internal point of line segment $uw$, and adding straight-line edges from $p$ to the four vertices of the quadrilateral.
See Figure~\ref{fig:subdivision}.  

\begin{figure}[htb]
\centering
\hfill
\subfigure[]{\includegraphics[scale=1.5]{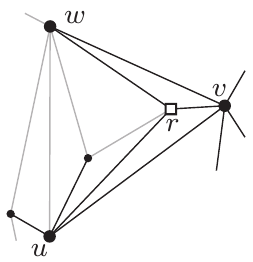}\label{fig:subd_pre}}\hfill
\subfigure[]{\includegraphics[scale=1.5]{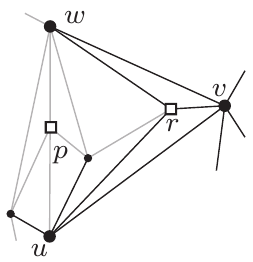}\label{fig:subd_post}}\hfill
\caption{If edge $uw$ has been added when triangulating the drawing (a), then we subdivide $uw$ with the insertion of a vertex $p$ and we triangulate the two faces $p$ is incident to (b).}
\label{fig:subdivision}
\end{figure}

After convexifying the quadrilateral $urwv$, we remove the temporary triangulation edges, remove $r$, and add the edge $uw$. See Figure~\ref{fig:v_edge}.  

Thus with two calls to Quadrilateral Convexification (one for each drawing), we have added one  edge to both drawings. Moreover, the two drawings are still topologically equivalent drawings of the same planar graph. We can continue until every facial cycle of every connected component (except an isolated vertex or edge) is a triangle.  Since planar graphs have $O(n)$ edges, 
in total we use $O(n)$ calls to Quadrilateral Convexification.  

\medskip\noindent
{\bf Part B.} At this point every connected component is either an edge, a vertex, or a triangulation.  We will now morph and add edges to connect the components together. Once all of the components have been connected we will again appeal to Part A to triangulate the full connected graph.

Consider a face that has $c'$ connected components drawn inside it.  The outer boundary of the face is a triangle $abc$, and each inner component is either a single vertex or edge, or is bounded by a triangle.
The high-level idea is to shrink all of the inner components so that there is freedom to move them around within $abc$. We will use this freedom to line up the inner components in the same order in both drawings, starting from a position near $a$. It will then be straightforward to connect the outer component with the first inner component via an edge from $a$ and add edges between consecutive inner components to connect the rest.

In order to determine the size to which we must shrink the inner components, it will be helpful to have triangular boundaries for each inner component (even for the ones that are single vertices or edges) that are disjoint, lie within $abc$, and have positive area. 


To this end, consider each inner component that is a vertex $u$. In each drawing a sufficiently small disk centered at $u$ contains no part of the drawing other than vertex $u$ itself. Then a positive-area triangle 
formed by vertex $u$ and two new ``dummy" vertices
can be inserted in the disk to become the triangular boundary for that component. 
Similarly, consider each inner component that is an edge $uv$. In each drawing a sufficiently thin rectangle having $uv$ as a side contains no part of the drawing other than edge $uv$ itself. Then a positive-area triangle 
formed by edge $uv$ and a new ``dummy'' vertex
can be inserted in the rectangle to become the triangular boundary for that component.

Now, consider a uniform triangular grid such that $a$ lies on a vertex of the grid and each grid cell is either a homothetic copy of $abc$ (called an \emph{upward} cell) or a homothetic copy of the triangle obtained by rotating $abc$ by $\pi$ radians (called a \emph{downward} cell). Let every second row of cells in any of the three directions determined by the sides of $abc$ be a \emph{road}. Let any cell that is not in any road be a \emph{home}. We fix the parities of the rows such that the upward cell adjacent to $a$ in $abc$ is a home. Note that all homes are thus upward cells. See Figure~\ref{fig:partb:homes}. We have yet to specify the size of a grid cell. We require a grid size such that the following conditions are satisfied: 1) there are at least $c'$ homes along $ac$ and 2) each bounding triangle of an inner component contains a home.

\begin{figure} 
	\centering
	\includegraphics[scale=0.5]{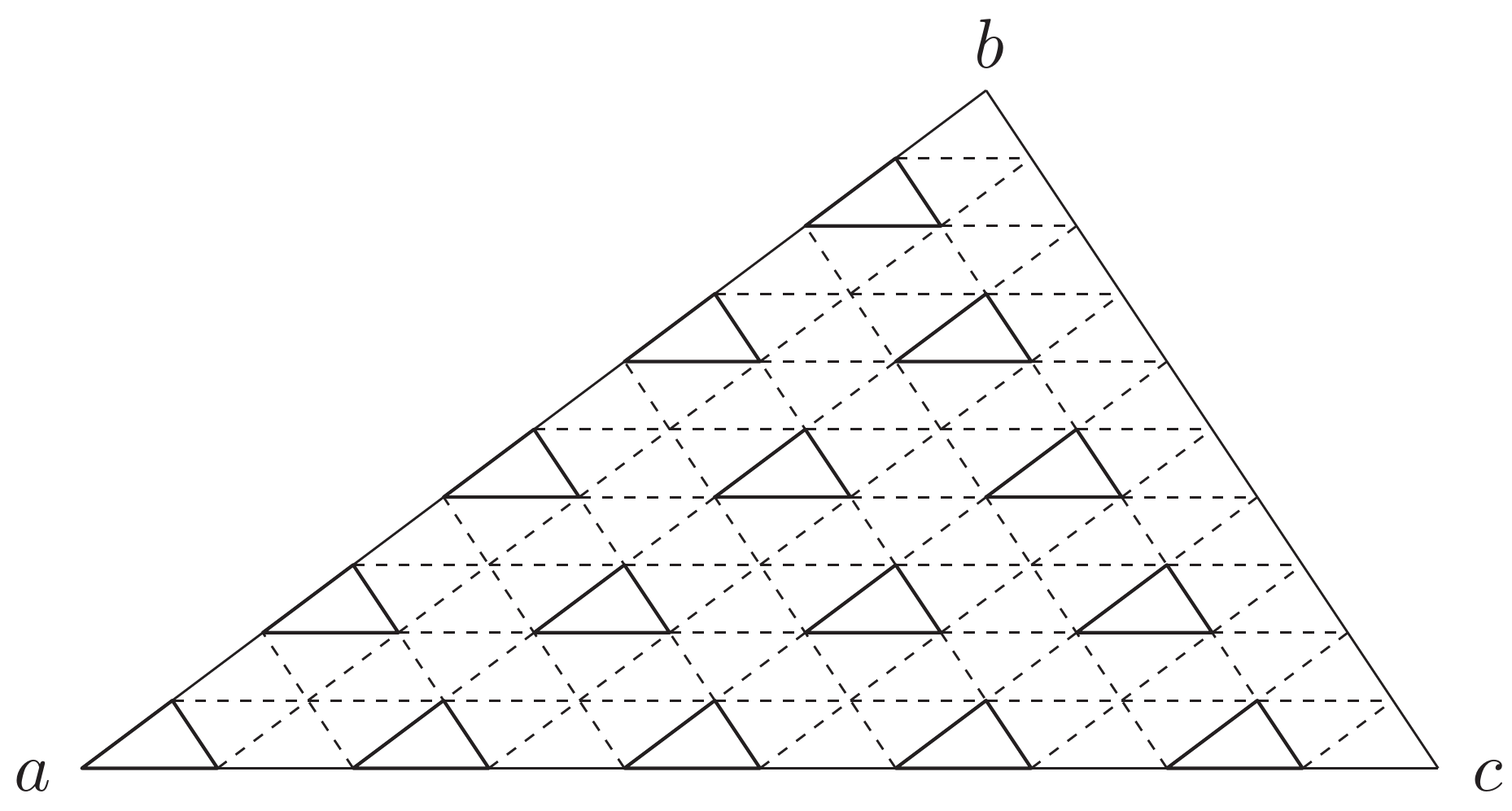}
	\caption{Roads (bounded by dashed lines) and homes (small solid triangles) within triangular face $abc$.}
	\label{fig:partb:homes}
\end{figure}

For each inner component, let its \emph{home region} be the largest homothetic copy of an upward cell that can be inscribed in the component's triangular boundary. Let $d$ be the diameter of $abc$. We can, for example, set the grid size such that the diameter of a grid cell is the minimum of: 1) $d/(2c')$ and 2) a fifth of the diameter of the smallest home region. If the diameter of a grid cell is at most the first value, then there must be at least $2c'$ upward cells along $ac$, half of which are homes. If the diameter of a grid cell is at most the second value, then every home region is at least five times the size of an upward cell. Every home region of this size must contain a home. 
%
%
In fact, every home region of this size contains a homothetic copy of $abc$ whose side lengths are three times the side lengths of an upward cell; further, at least one of the six upward cells contained in this homothetic copy of $abc$ is a home. 

\remove{
Say an edge of the home region lies in a row $r$. Then the intersection of the home region and the next non-road row after $r$ is long enough to accommodate at least three consecutive upward cells. Since the home region is not necessarily aligned to the triangular grid, we only have the guarantee that at least two consecutive upward cells are contained in the intersection. One of two adjacent upward cells in a non-road row must be a home. See Figure~\ref{fig:partb:home-region}.

\begin{figure} 
	\centering
	\includegraphics[scale=0.5]{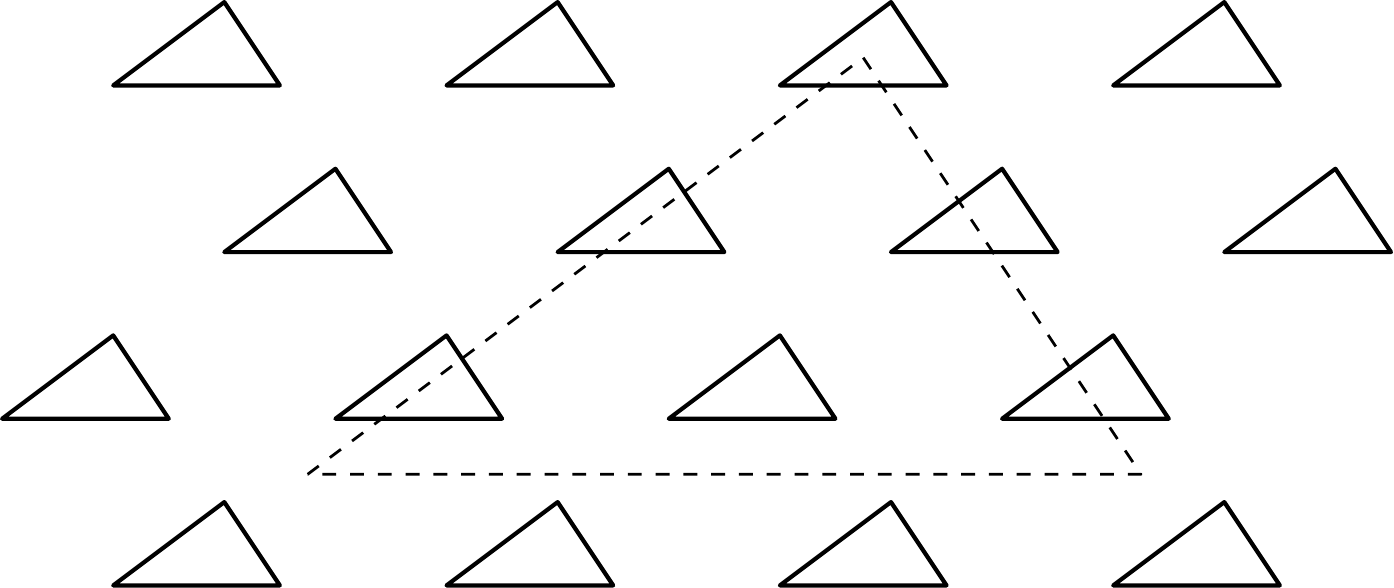}
	\caption{A home region (dashed) five times the size of a home (solid) must contain a home.}
	\label{fig:partb:home-region}
\end{figure}
}

Consider an inner component bounded by triangle $uvw$ with a home region that contains a home $h$. We choose a non-dummy vertex of $uvw$ and call it the component's \emph{connector}. Our goal is to morph the inner component into $h$ such that its connector is near the vertex of $h$ that is farthest from $ac$. We will do so with $O(1)$ morphs of $uvw$ via linear movements of single vertices of $uvw$, during which the contents of $uvw$ move along linearly. By Lemma~\ref{lemma:convex-comb}, these are each unidirectional morphs. 

We first morph $uvw$ to occupy exactly the component's home region with a linear movement of each vertex. We then rotate the positions of the vertices of $uvw$ so that the connector is in the position farthest from $ac$. This can be done with six linear movements of the vertices, using the home region's midpoint polygon as an intermediate. Finally, with a linear movement of each vertex, we shrink $uvw$ into $h$ such that $uvw$ becomes a slightly smaller homothetic copy of $h$ and does not intersect the boundary of $h$. There is no interference between inner components since $uvw$ remains in its original bounding triangle for the duration of the morph. See Figure~\ref{fig:partb:shrink}. Observe that we performed $O(1)$ unidirectional morphs for each of the $c'-1$ connected components.

\begin{figure} 
	\centering
	\includegraphics[scale=0.7]{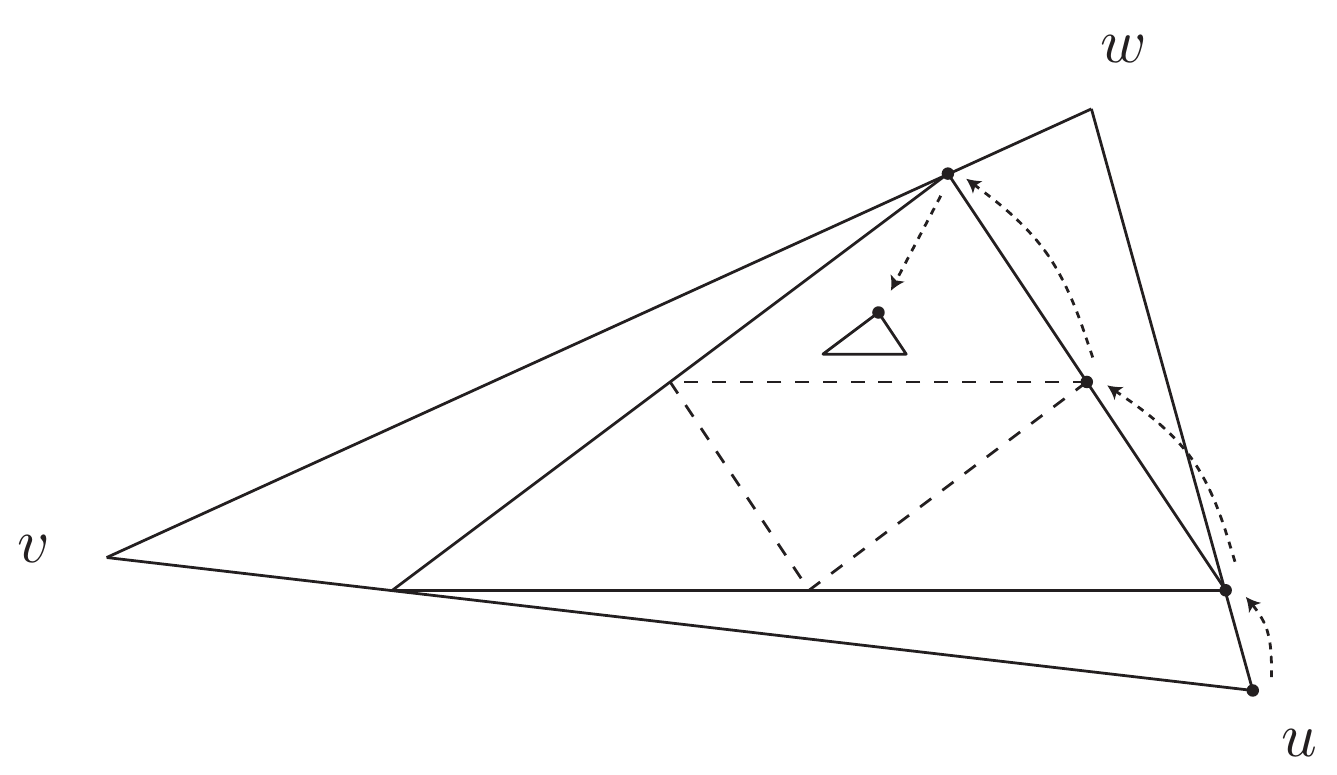}
	\caption{Movements of the connector of $uvw$ used to shrink and rotate $uvw$ into $h$. The solid triangles are, from outermost to innermost: $uvw$, the home region of $uvw$, the final position of $uvw$ within $h$. The dashed triangle is the midpoint polygon of the home region of $uvw$.}
	\label{fig:partb:shrink}
\end{figure}

We can now arrange the $c'$ inner components in any arbitrary order in the homes along edge $ac$ using $c'$ home swaps. Since all of the inner components are in homes, the road network is unimpeded and the roads are wide enough to allow passage of any inner component. By navigating along the road network, it is straightforward to swap the homes of two inner components using $O(1)$ linear movements of the components. Thus, in $O(c')$ unidirectional morphs, we can arrange the inner components into the homes along $ac$ in the same order in both drawings. Finally we add edges between the connectors of adjacent components and add an edge from $a$ to the connector of the first component and we are done connecting the components. See Figure~\ref{fig:partb:connected}.

\begin{figure} 
	\centering
	\includegraphics[scale=0.8]{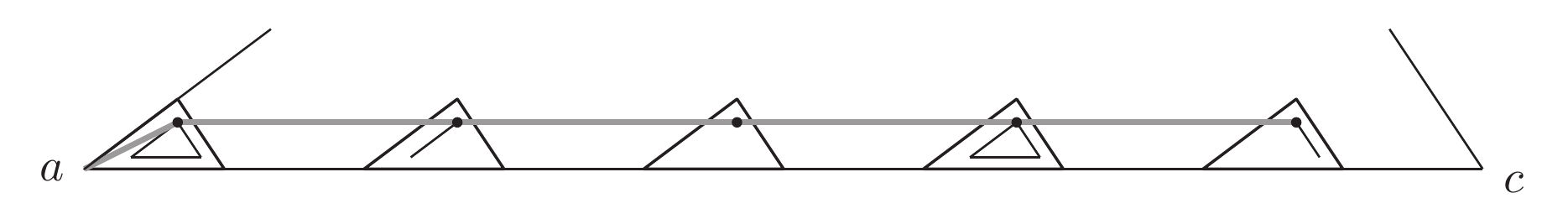}
	\caption{Edges (thick grey lines) linking the connectors (dots) of the inner components that lie in the homes along $ac$.}
	\label{fig:partb:connected}
\end{figure}

Once all components are connected, we remove all dummy vertices and edges and appeal to part A to triangulate the full connected graph. This completes the description of part B.

\medskip
Altogether, we use $O(c)$ unidirectional morphs to connect the $c$ components of the graph. We also add $O(n)$ edges in Part A, and each addition takes two calls to Quadrilateral Convexification.   
Furthermore, each call to Quadrilateral Convexification uses $O(1)$ unidirectional morphing steps by Theorem~\ref{thm:quad-convex}.  Therefore, in total we use $O(n)$ unidirectional morphing steps.

We now analyze the running time of the algorithm.  We call Quadrilateral Convexification $O(n)$ times, and each call takes $O(n^2)$ time.  This dominates the other work done by the algorithm (e.g.~we find $O(n)$ temporary triangulations, each of which takes
$O(n \log n)$ time).
Thus the total running time is $O(n^3)$.    
\end{proof}

\section{Morphing Between Two Triangulations}
\label{sec:morph-triangulations}

In this section we prove our main result, Theorem~\ref{th:main}, for the case of triangulations.
Since the previous section reduced the general case to the case of triangulations, this will complete the proof of Theorem~\ref{th:main}.

\begin{theorem}
Let $\Gamma_1$ and $\Gamma_2$ be two triangulations that are topologically equivalent drawings of an $n$-vertex maximal planar graph $G$. There is a morph from $\Gamma_1$ to $\Gamma_2$ that is composed of $O(n)$ unidirectional morphs.  More precisely, the morph uses $O(n)$ unidirectional morphs plus $O(n)$ calls to Quadrilateral Convexification.  Furthermore, the morph can be constructed in $O(n^3)$ time.
\end{theorem}  

As described in Subsection~\ref{sec:main-result}, our approach is to find an internal vertex $v$ of degree at most five and then morph the drawings using $O(1)$ unidirectional morphs so that $v$ can be contracted to the same neighbor $u$ in both drawings.  We will prove a time bound of $O(n^2)$ for this step.  More details are given below.  

After this step, we contract $v$ to $u$ in both drawings and apply recursion to find a morph $\cal M$ between the two smaller triangulations.  The result is a pseudomorph:  we contract $v$ to $u$ in $\Gamma_1$, apply $\cal M$, and then reverse the contraction of $v$ to $u$ in $\Gamma_2$. 
The number of steps satisfies the recurrence relation $S(n) = S(n-1) + O(1)$, which solves to $S(n) \in O(n)$.  
Finally, we apply Theorem~\ref{thm:pseudomorph} to convert the pseudomorph to a morph with the same number of unidirectional morphs.
Each appeal to Theorem~\ref{thm:pseudomorph} takes time $O(n)$ since that is the bound on the number of unidirectional morphs.  Thus the total time bound is given by the recurrence relation $T(n) = T(n-1) + O(n^2)$, which solves to $T(n) \in O(n^3)$.

We now fill in the details about finding vertex $v$ and contracting it to a common neighbor in both drawings. If there are no internal vertices, then we only have the outer triangle and a unidirectional morph with $O(1)$ morphing steps is easily computed in time $O(1)$. Otherwise, we claim that there is an internal vertex $v$ of degree at most 5.  To see this, note that the triangulation has $3n-6$ edges, and thus the  sum of the degrees is $6n-12$.  Every vertex has degree at least 3, and if all the internal vertices had degree at least 6, the sum of the degrees would be at least $3(3) + 6(n-3) = 6n - 9$.

If $v$ has degree 3 then we can contract $v$ to the same neighbor $u$ in both  drawings.  
If $v$ has degree 4 or 5 then let $u$ be a neighbor of $v$ to which $v$ can be contracted in $\Gamma_2$.  (The existence of $u$ is guaranteed by applying Lemma~\ref{lem:5-kernel} to the polygon $\Delta(v)$ in $\Gamma_2$.)   
If $v$ can be contracted to $u$ in $\Gamma_1$ we are done.  Otherwise, we will 
morph $\Gamma_1$ to make this possible.  Note that there are no external chords from $u$ to any vertex of $\Delta(v)$, because the two drawings are topologically equivalent, and $\Gamma_2$ has no such chords.

If $v$ has degree 4 then the neighbors of $v$ must form a non-convex quadrilateral $Q$ with vertices $abcd$, where $d$, say, is the reflex vertex. We contract $v$ to $d$ and apply Quadrilateral  Convexification to $Q$, noting that the quadrilateral has no external chords.  Specifically, there is no chord $ac$ because $v$ cannot be contracted to $u$, so $u$ must be one of $a, c$, and, as noted above, there are no external chords incident to $u$.   
After convexification, we move $v$ slightly into the interior of $\Delta(v)$ to obtain a drawing in which vertex $v$ can be contracted to $u$.   The result is a pseudo-morph, which we convert to a morph using Theorem~\ref{thm:pseudomorph}.  The number of unidirectional steps is $O(1)$ and the time bound is $O(n^2)$ by Theorem~\ref{thm:quad-convex}. 

If $v$ has degree 5 then the neighbors of $v$ form a pentagon $abcde$ where we want to contract $v$ to $u=a$, say.  We use a similar method, 
morphing the pentagon until it is ``almost'' convex by making a few calls to Quadrilateral Convexification.  More formally:

\begin{lemma} 
Let $\Gamma$ be an $n$-vertex triangulation. Suppose $v$ is a non-boundary vertex of degree 5, with neighbors forming a pentagon $abcde$, and suppose that $a$ is not incident to any external chords of the pentagon. 
Then we can morph $\Gamma$, keeping the outer boundary fixed, so that $v$ can be contracted to $a$.
The morph consists of at most two calls to Quadrilateral Convexification plus a constant number of unidirectional morphing steps.   Furthermore, the morph can be found in $O(n^2)$ time.
\end{lemma}

\begin{proof}
Because $v$'s neighbors form a pentagon, $v$ can be contracted to some neighbor $x$ while preserving planarity (Lemma~\ref{lem:5-kernel}).  
If $x=a$ we are done, so suppose they are different.
Contract $v$ to $x$.
We will divide into two cases depending on whether $x$ is adjacent to $a$ on the pentagon or not.    

First consider the case where $x$ is not adjacent to $a$.  Without loss of generality, suppose that $x=c$.
See Figure~\ref{fig:pentagon-case1}.
We use Quadrilateral Convexification to convexify $acde$.  This is possible since $ce$ is an inner chord and $ad$ is not an external chord.   Now $a$ sees all the quadrilateral $acde$ as well as the triangle $abc$.  Thus $v$ can now be contracted to $a$.

\begin{figure}[htb]
\centering
\def\sf{1.5}
\hfill\subfigure[]{\includegraphics[scale=\sf]{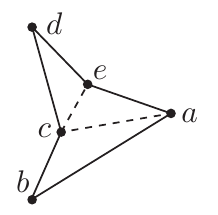}\label{fig:pentagon-case1}}\hfill
\subfigure[]{\includegraphics[scale=\sf]{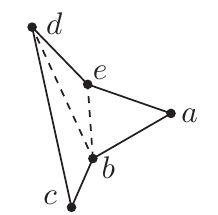}\label{fig:pentagon-case2}}\hfill
\caption{\subref{fig:pentagon-case1}~Vertex $v$ has been contracted to $x=c$ and we wish to contract $v$ to $a$. 
\subref{fig:pentagon-case2}~Vertex $v$ has been contracted to $x=b$ and we wish to contract $v$ to $a$. 
}
\end{figure}

Next consider the case where $x$ is adjacent to $a$.  Without loss of generality, suppose that $x=b$. 
See Figure~\ref{fig:pentagon-case2}.  Use Quadrilateral Convexification to convexify $abde$.  This is possible since $be$ is an inner chord and $ad$ is not an external chord.  Now $d$ sees all the quadrilateral $abde$ as well as the triangle $bcd$.  Move vertex $v$ in a straight line from $x=b$ to $d$.   This puts us in the first case.  We use Theorem~\ref{thm:pseudomorph} to convert the resulting pseudomorph to a morph.

The morph consists of at most two calls to Quadrilateral Convexification plus a constant number of unidirectional morphing steps.   The morph can be found in $O(n^2)$ time.
\end{proof}

\section{Quadrilateral Convexification}
\label{sec:quad-convexification}

In this section we give an algorithm for Problem~\ref{prob:quad-convex}, to 
morph a triangulation in order to convexify a quadrilateral. The main result of the section is as follows:

\rephrase{Reminder of Theorem~\ref{thm:quad-convex}}{\quadconvex}

In the rest of this section we will prove Theorem~\ref{thm:quad-convex}. Recall that we have an $n$-vertex triangulation $\Gamma$ and a quadrilateral $abcd$ in $\Gamma$ with no vertex inside it and such that neither $ac$ nor $bd$ is an edge outside of $abcd$. Our goal is to morph $\Gamma$ so that $abcd$ becomes convex using a single unidirectional morph.

We first describe the main idea. If $abcd$ is convex, then we are done. Assume without loss of generality that $d$ is the reflex vertex of the non-convex quadrilateral, i.e., the angle at $d$ internal to quadrilateral $abcd$ is larger than $\pi$ radians. This implies that $b$ is the tip of the arrowhead shape and that the triangulation contains the edge $bd$ (see Figure~\ref{fig:4-gon}). Change the frame of reference so that $bd$ is ``almost'' horizontal. We will use one unidirectional morph that moves vertices along horizontal lines, i.e.~preserving their $y$-coordinates.  Our main tool will be a result about re-drawing a plane graph to have convex faces while keeping all vertices at the same $y$-coordinate---this is a result by Hong and Nagamochi~\cite{hn-cdhpgcpg-10} expressed in terms of level planar drawings of hierarchical graphs.  To complete the proof we will show that the linear motion from the original drawing to the new drawing in fact preserves planarity, and therefore is a unidirectional morph.

\begin{figure}[h]
\centering
\includegraphics[width=0.4\textwidth]{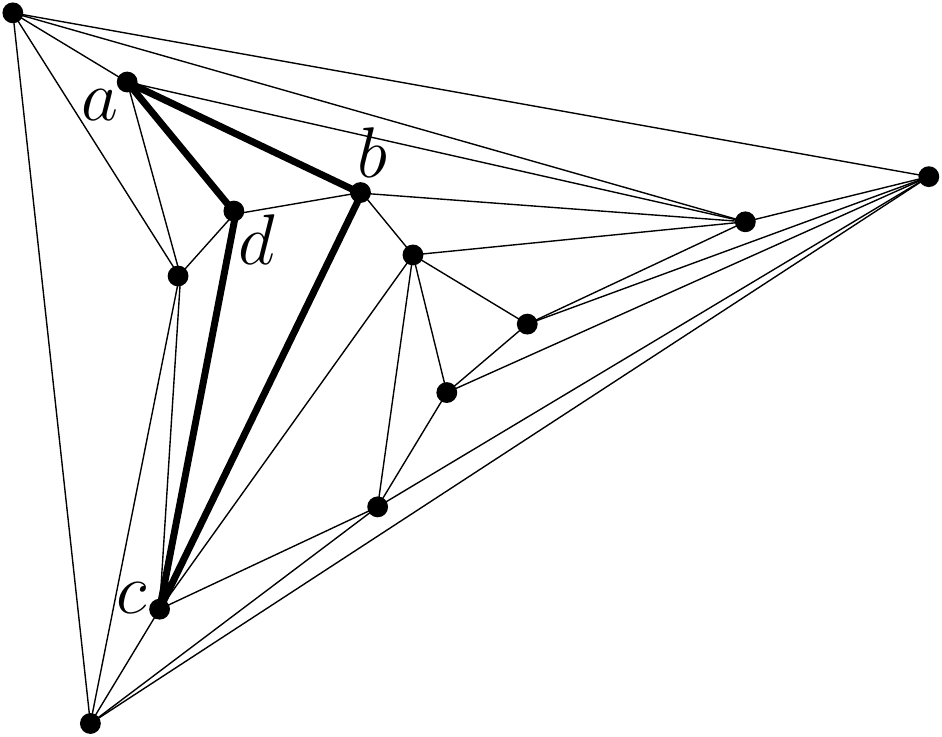}
\caption{Triangulation $\Gamma$, containing a non-convex quadrilateral $abcd$ with a reflex vertex $d$, with no vertex inside it, with edge $bd$ inside it, and with no chord outside it.}
\label{fig:4-gon}
\end{figure}

We introduce some definitions and terminology that are needed for Hong and Nagamochi's result. A hierarchical graph is a graph with vertices assigned to layers.  More formally,
a {\em hierarchical graph} is a triple $(G,L,\gamma)$ such that $G$ is a graph, $L$ is a set of horizontal lines (sometimes called {\em layers}), and $\gamma$ is a function mapping each vertex of $G$ to a line in $L$ in such a way that, if an edge $(u,v)$ belongs to $G$, then $\gamma(u)\neq \gamma(v)$. (It is conventional to assume the layers are $1, 2, \ldots k$, i.e.,~that successive horizontal lines are distance 1 apart, but that turns out to be unnecessary for Hong and Nagamochi's result, as the fact that the horizontal lines are equally spaced is nowhere used in their proof.) With the lines in $L$ ordered from bottom to top, $\gamma$ represents a partial order on the vertices of $G$, and we write  $u\prec_{\gamma} v$ if the line $\gamma(u)$ is below the line $\gamma(v)$.

A {\em level drawing} of a hierarchical graph $(G,L,\gamma)$ maps each vertex $v$ of $G$ to a point on the line $\gamma(v)$ and each edge to a $y$-monotone curve. A {\em level planar drawing} is a level drawing in which no two curves representing edges cross; such a drawing is \emph{straight-line} if edges are drawn as straight-line segments, and \emph{convex} if faces are delimited by convex polygons.

In our situation, we have a straight-line level planar drawing of a hierarchical graph and we want a straight-line convex level planar drawing that ``respects'' the embedding.  
More abstractly and more generally, Hong and Nagamochi define a
{\em hierarchical plane graph} to be a hierarchical graph together with a combinatorial embedding (a rotation system and outer face) corresponding to some level planar drawing.  In other words, a hierarchical plane graph is an equivalence class of drawings, and a drawing of a hierarchical plane graph must be a member of the equivalence class.


In order to guarantee a convex level planar drawing, Hong and Nagamochi require some conditions. Given a hierarchical plane graph $(G,L,\gamma)$, an {\em st-face} of $G$ is a face delimited by two paths $(s=u_1,u_2,\dots,u_k=t)$ and $(s=v_1,v_2,\dots,v_l=t)$ such that $u_i \prec_{\gamma} u_{i+1}$, for every $1\leq i\leq k-1$, and such that $v_i\prec_{\gamma} v_{i+1}$, for every $1\leq i\leq l-1$.  We say that $(G,L,\gamma)$ is a {\em hierarchical plane st-graph} if every face of $G$ is an st-face. 
Hong and Nagamochi give an algorithm that constructs a convex straight-line level planar drawing of any hierarchical plane st-graph~\cite{hn-cdhpgcpg-10}.
Here we explicitly formulate a weaker version of their main theorem.\footnote{We make some remarks. First, the result in~\cite{hn-cdhpgcpg-10} proves that a convex straight-line level planar drawing of $(G,L,\gamma)$ exists even if a convex polygon representing the cycle delimiting the outer face of $G$ is arbitrarily prescribed. More precisely, Hong and Nagamochi show a necessary and sufficient condition for a convex polygon (possibly with flat angles) to be the polygon delimiting the outer face of a convex straight-line level planar drawing of $(G,L,\gamma)$. However, this condition is always satisfied if every internal angle of the polygon is convex. Second, the result holds for a super-class of the $3$-connected planar graphs, namely for all the graphs that admit a convex straight-line drawing~\cite{cyn-lacdp-84,t-prg-84}. 
}

\begin{theorem} \label{th:hong-nagamochi} (Hong and Nagamochi~\cite{hn-cdhpgcpg-10})
Every $3$-connected hierarchical plane st-graph $(G,L,\gamma)$ admits a convex straight-line level planar drawing.
\end{theorem}


We now proceed to prove Theorem~\ref{thm:quad-convex} and hence to solve the Quadrilateral Convexification problem with a single unidirectional morph. First, we rotate the frame of reference so that edge $bd$ is horizontal and then we rotate it a bit more, so that no two vertices have the same $y$-coordinate and so that the strip delimited by the horizontal lines through $b$ and $d$ contains no vertex of $\Gamma$, except for the presence of $b$ and $d$ on its boundary. Remove edge $bd$ from $\Gamma$ obtaining a planar straight-line drawing $\Gamma'$ of a planar graph $G'$. Draw a horizontal line through each vertex; let $L$ be the set of these lines and let $\gamma$ be the function that maps each vertex to the unique line in $L$ through it; observe that, by the assumptions, $\gamma$ represents a total order on the vertices of $G'$. We have the following. 

\begin{lemma} \label{le:orientation}
$(G',L,\gamma)$ is a $3$-connected hierarchical plane st-graph.
\end{lemma}
\begin{proof}
By construction, $\Gamma'$ is a straight-line level planar drawing of $(G',L,\gamma)$, hence $(G',L,\gamma)$ is a hierarchical plane graph. 

Further, every face of $G'$ is an st-face. This is trivially true for all faces delimited by triangles in $\Gamma'$, since $\gamma$ represents a total order on the vertices of $G$. Moreover, this is true for the face delimited by $abcd$, since by the choice of the reference frame $a$ is below both horizontal lines $\gamma(b)$ and $\gamma(d)$ and $c$ is above both of them, or vice versa. It follows that $(G',L,\gamma)$ is a hierarchical plane st-graph.





\begin{figure}[htb]
 \centering 
\includegraphics[scale=1]{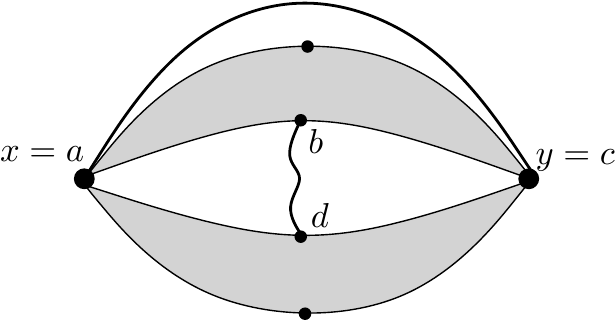}\label{fig:3-connected-proof-2components}
 \caption{Illustration for the proof of Lemma~\ref{le:orientation}.}
 \label{fig:3-connected-proof}
\end{figure}


Finally, we prove that $G'$ is $3$-connected. Refer to Figure~\ref{fig:3-connected-proof}. Suppose, for a contradiction, that $G'$ is not $3$-connected. Then there exist two vertices $x$ and $y$ whose removal from $G'$ results in a disconnected graph $G''$. Since $\Gamma$ is $3$-connected, adding the edge $bd$ to $G''$ reconnects this graph. Therefore, $b$ and $d$ must be in different components of $G''$. Since $G'$ contains paths $bad$ and $bcd$, it follows that $\{x,y\} = \{a,c\}$.  Since removing $a,c$ disconnects $G'$, there is another face of $G'$ that contains $a$ and $c$.  This face is delimited by three edges (as the one delimited by quadrilateral $abcd$ is the only face of $G'$ which is not triangular), therefore contains the edge $ac$. However, this implies that $xy=ac$ is a chord external to polygon $abcd$ in $\Gamma$, thus contradicting the assumptions.
\end{proof}

By Lemma~\ref{le:orientation}, $(G',L,\gamma)$ is a $3$-connected hierarchical plane st-graph. By Theorem~\ref{th:hong-nagamochi}, a convex straight-line level planar drawing $\Lambda'$ of $(G',L,\gamma)$ exists. In particular, polygon $abcd$ is convex in $\Lambda'$. Construct a straight-line planar drawing $\Lambda$ of $G$ from $\Lambda'$ by drawing edge $bd$ as an open straight-line segment. Due to the initial choice of the reference frame, this introduces no crossing in the drawing. To solve Problem Quadrilateral Convexification we use a single linear morph, $\langle \Gamma,\Lambda \rangle$, from $\Gamma$ into $\Lambda$. See Figure~\ref{fig:convexifier-morphing}. We now prove the following lemma.

\begin{figure}[htb]
 \centering 
 \includegraphics[width = 0.8\textwidth]{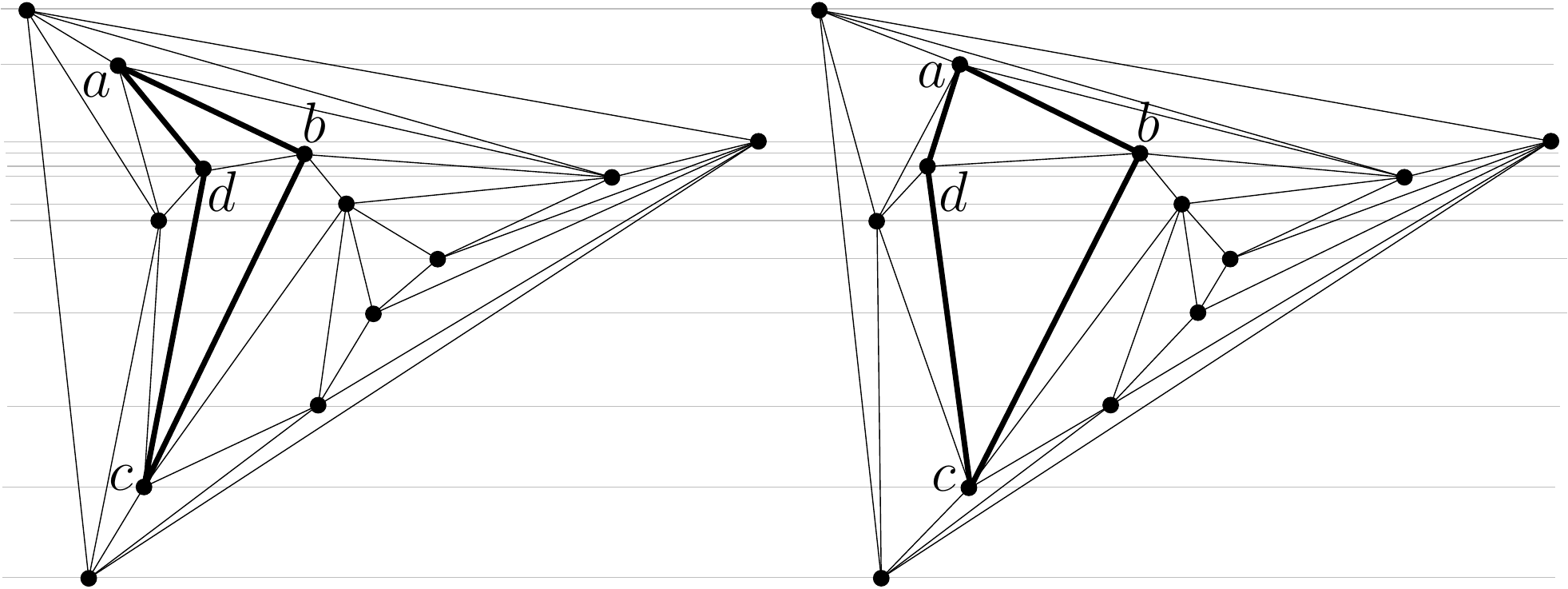}
 \caption{The linear morph $\langle \Gamma,\Lambda \rangle$.}
 \label{fig:convexifier-morphing}
\end{figure}

\begin{lemma} \label{le:unidirectional}
The linear morph \mmorph{\Gamma,\Lambda} is planar and unidirectional.
\end{lemma}
\begin{proof} 
The morph is certainly unidirectional, since each vertex is on the same horizontal line in the initial drawing $\Gamma$ and the final drawing $\Lambda$.  To prove planarity we will argue that no vertex crosses an edge during the motion.  The tool we need is Lemma~\ref{lem:move_boundary}, which proves that if two points $p$ and $q$ each move at uniform speed along a horizontal line and $q$ is to the right of $p$ in their initial and final positions, then $q$ is to the right of $p$ at every time instant. Lemma~\ref{lem:move_boundary} will be stated and proved in Section~\ref{sec:sectors}. Thus it suffices to show that for each horizontal line $\ell$ in $L$, the left-to-right ordering of vertices lying on $\ell$ and points where edges cross $\ell$ is the same in $\Gamma$ as in $\Lambda$.   This follows directly from the fact that both drawings have the same faces, and every face is an st-face.  
\end{proof}

In order to complete the proof of Theorem~\ref{thm:quad-convex}, it remains to discuss the running time of the algorithm that solves Quadrilateral Convexification. Removing and re-inserting edge $bd$ takes $O(1)$ time. Other than that, we only need to apply Hong and Nagamochi's algorithm, which takes $O(n^2)$ time~\cite{hn-cdhpgcpg-10}. Thus the total running time is $O(n^2)$.

\section{Converting a Pseudo-Morph to a Morph}\label{se:geometry}
\label{sec:geometry}

In this section, we show how to convert a pseudomorph consisting of unidirectional morphing steps into a true morph of unidirectional morphing steps, assuming that the graph is triangulated and the vertices we contract have degree at most 5.
%
Specifically, we prove the following.

\rephrase{Reminder of Theorem~\ref{thm:pseudomorph}}{\pseudomorph}

Suppose that the given pseudo-morph  consists of the contraction
of an internal vertex $v$ with $\Deg(v)\le 5$ to a vertex $a$, followed by a
morph 
 $\mm=\langle\Gamma_{1},\ldots,\Gamma_{k+1}\rangle$ of the reduced
graph, and then the uncontraction of $v$ from $a$. Suppose that $\langle \Gamma_{i},\Gamma_{i+1}\rangle$ is an
$L_{i}$-directional morph, for $1\leq i\leq k$.
We assume that the drawings $\Gamma_{1},\ldots,\Gamma_{k+1}$ are given to us, and show how to update the sequence of drawings to those of $\cal M$ in time $O(n+k)$.

Specifically, we will show how to add $v$ and its incident edges back into each drawing of the morph $\cal M'$ keeping each step unidirectional.
We will preserve planarity by placing $v$ at an interior point of
the kernel of $\Delta(v)$.  Call the resulting morph $\cal M''$.
We will perform the modifications from $\cal M'$ to $\cal M''$ in time $O(k)$.
To obtain the final morph $\cal M$, we
replace the original contraction of $v$ to $a$ by a unidirectional morph that
moves $v$ from its initial position to its position at the start of
$\cal M''$, then follow the steps of $\cal M''$, and then replace the
uncontraction of $v$ by a unidirectional morph that moves $v$ from its
position at the end of $\cal M''$ to its final position.  The result is a
true morph that consists of $k+2$ unidirectional morphing steps.
It takes $O(n)$ time to add the two extra unidirectional morphs to the sequence, since we must add two drawings of an $n$-vertex graph. 

Thus our main task is to modify the morph $\cal M'$  by adding vertex $v$ and its incident edges back into each drawing of the morph sequence in constant time per drawing,  preserving planarity and 
maintaining the property that each step of the morph sequence is unidirectional.
We can ignore everything outside the polygon $P = \Delta(v)$. 
Note that vertex $a$ remains in the kernel of $P$ throughout the morph.
We distinguish cases depending on the degree of $v$.  
Section~\ref{sec:triangles_and_quads} proves Theorem~\ref{thm:pseudomorph} for the easy case where the degree of $v$ is 3 or 4.  Section~\ref{sec:penta} proves Theorem~\ref{thm:pseudomorph} for the more complicated case where $v$ has degree 5.



\remove{
This section is structured as follows. 
Section~\ref{sec:triangles_and_quads} deals with the case where $P$ has 3 or 4 vertices.
We consider separate cases depending on the degree of vertex $v$. For the case where
$\Deg(v)\in \{3,4,5\}$ we assume the neighborhood of $v$ is triangulated.  
In Section~\ref{sec:triangles_and_quads} we show how to
obtain the true morph from $\MM$ for the case where
$\Deg(v)\in\{3,4\}$. The case where $\Deg(v)=5$ is more involved and
it is treated separately in Section~\ref{sec:penta}.  
}


\subsection{Vertex \boldmath{$v$} of Degree 3 or 4}
\label{sec:triangles_and_quads}

In this section we prove Theorem~\ref{thm:pseudomorph} for the case where the contracted vertex $v$ has degree 3 or 4, i.e.~$P$ is a triangle or quadrilateral.
\remove{
\begin{lemma}\label{lem:tri_and_quad}
  Let $\MM$ be a planar pseudo-morph from $\Gamma$ to $\Gamma'$ that
  contracts vertex $v$ and consists of $O(k)$ unidirectional steps.
  Suppose that that $\Deg(v)\in\{3,4\}$ with the neighborhood of $v$
  being triangulated. Then we can obtain a planar morph $M$ from
  $\Gamma$ to $\Gamma'$ consisting of the same number of
  unidirectional steps in $O(k)$ time.
\end{lemma}
\begin{proof}
}
%
  If $P$ is
  a triangle then by Lemma~\ref{lemma:convex-comb} we can place $v$ at
  a fixed convex combination of the triangle vertices in all the
  drawings $\Gamma_i$.

  If $P$ is a quadrilateral $abcd$ then the line segment $ac$ stays in
  the kernel of $P$ because vertex $a$ stays in the kernel of $P$.  Thus, we can place $v$ at a fixed convex combination of $a$
  and $c$ in all the drawings $\Gamma_i$ (using the degenerate version
  of Lemma~\ref{lemma:convex-comb} where the triangle collapses to a
  line segment).
  
 The coordinates of $v$ in each $\Gamma_i$ can be computed in constant time, so the total time bound is $O(k)$. 
%


\subsection{Vertex \boldmath{$v$} of Degree 5}
\label{sec:penta}

In this section we prove Theorem~\ref{thm:pseudomorph} for the case where the contracted vertex $v$ has degree 5.
Our goal will be to place vertex $v$ very close to the vertex $a$ to which it was contracted in the pseudo-morph.
\remove{
In this section we show how to obtain a true planar morph from a
planar pseudomorph $\MM$ where vertex $v$ of degree 5 is contracted
and consists only of unidirectional steps. We state this formally in
the following lemma.

\begin{lemma}\label{lem:penta}
  Let $\MM$ be a planar pseudo-morph from $\Gamma$ to $\Gamma'$ that
  contracts vertex $v$ and consists of $O(k)$ unidirectional steps.
  Suppose that that $\Deg(v)=5$ with the neighborhood of $v$ being
  triangulated. Then we can obtain a planar morph $M$ from $\Gamma$ to
  $\Gamma'$ consisting of the same number of unidirectional steps in
  $O(k)$ time.
\end{lemma}
}
%
Let $P= \Delta(v)$ be the pentagon $abcde$ labelled clockwise. 
We use the notation that $b$ is at point $b_i$ in drawing $\Gamma_i$, etc.

We may assume that vertex $a$ is fixed throughout the entire morph. 
%
This is not a loss of generality because if vertex $a$ moves, we can
translate the whole drawing to move it back: the morph in which every vertex $v$ moves $k_v$ units along direction $\bar \ell$ is planar if and only if the morph in which every vertex moves $k_v-k_a$ units along direction $\bar \ell$ is planar; and note that vertex $a$ stays fixed in the latter morph.

We want to place $v$ within distance $\varepsilon$ of $a$.  
We want $\varepsilon$ small enough so that 
at any time instant $t$ during morph $\langle
\Gamma_1,\dots,\Gamma_{k+1} \rangle$ the intersection between the disk $D$ centered at $a$ with
radius $\varepsilon$ and the kernel of polygon $P$ consists of a
positive-area sector $S$ of $D$.  
Since the morph consists of $k$ linear morphs, we can compute 
a value for $\varepsilon$ as follows. 
For $1\leq i\leq k$ let $\varepsilon_{i}$ be the minimum 
during the unidirectional morph $\langle\Gamma_{i},\Gamma_{i+1} \rangle$ of 
the distance from $a$ to any of the edges $bc, cd, de$.  
We claim that $\varepsilon_{i}$ can be computed in constant time on our real RAM model of computation.
It suffices to compute the minimum distance between $a$ and each of the three infinite lines through $bc, cd,$ and $de$.  This is because we can test if the minimum occurs inside the relevant line segment, and we can compute the minimum distance from $a$ to each of the moving endpoints $b,c,d,e$.
Consider the line through points $p,q$, where $pq = bc, cd,$ or $de$, and points $p$ and $q$ each move on parallel lines at uniform speed during the morph $\langle\Gamma_{i},\Gamma_{i+1} \rangle$.
As a function of time, $t$, the square of the distance from $a$ to the line through $pq$ is of the form $f(t)/g(t)$ where $f$ and $g$ are  quadratic polynomials in $t$.  We want to solve for the derivative being zero.  The derivative has a cubic polynomial in the numerator, and a cubic equation can be solved in constant time on a real RAM model.
%
%

After computing each $\varepsilon_i$, 
let $\varepsilon=\min_{1\leq i\leq k}\{\varepsilon_{i}\}$.
Then $\varepsilon$ can be computed in time $O(k)$.

Fix $D$ to be the disk of radius $\varepsilon$ centered at $a$.
In case $a$ is a convex vertex of $P$, the sector $S$ is bounded by the line segments $ab$ and $ae$ and we call it a \emph{positive} sector.  See Figure~\ref{fig:def-convex-disk}.
In case $a$ is a reflex vertex of $P$, the sector $S$ is bounded by the extensions of line segments $ab$ and $ae$ and we call it a \emph{negative} sector.  
See Figure~\ref{fig:def-reflex-disk}.
More precisely, let $b'$ and $e'$ be points so that
$a$ is the midpoint of the segments $bb'$ and $ee'$ respectively.  The negative sector is bounded by the segments $ae'$ and $ab'$.
Note that when an $L$-directional morph is applied to $P$, the points $b'$ and $e'$ also move at uniform speed in direction $L$.

\remove{ 
Observe that irrespective of whether $a$ is a convex or reflex vertex,
the sector $S$ is the intersection of $D$, the points that are to the
left of the line through $bb'$ (as traversed from $b$ to $b'$) and the
points that are to the right of the line through $ee'$ (as traversed
from $e$ to $e'$). This can be informally rephrased as the set of
points that are in $D$ and ``between'' lines $bb'$ and $ee'$. 
}

\begin{figure}[htb]
 \centering 
 \subfigure[]{\includegraphics[scale=1.7]{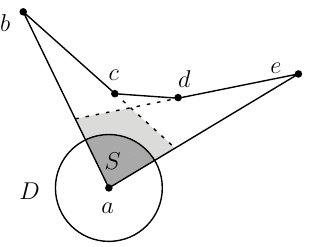}\label{fig:def-convex-disk}}
 \hspace{20pt}
 \subfigure[]{\includegraphics[scale=1.7]{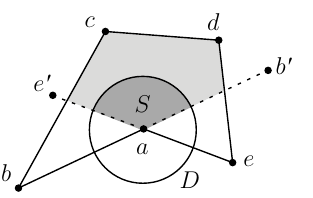}\label{fig:def-reflex-disk}}
 \caption{A disk $D$ centered at $a$ whose intersection with the
   kernel of $P$ (the lightly shaded polygonal region) is a non-zero-area
   sector $S$ (darkly shaded).  \subref{fig:def-convex-disk} Vertex $a$ is convex and $S$ is a
   positive sector.  \subref{fig:def-reflex-disk} Vertex $a$ is reflex and $S$ is a negative
   sector.}
 \label{fig:def-sectors}
\end{figure}

The important property we use from now on is that any point in the sector $S$ lies in the kernel of polygon $P$. This property immediately follows from the choice of $\varepsilon$. Let the sector in drawing $\Gamma_i$ be $S_i$, for $ i=1, \ldots,  k+1$.  
Recall our convention that 
$\langle \Gamma_i, \Gamma_{i+1} \rangle$ is an $L_i$-directional morph. 

Our task is to choose for each $i$ a position $v_i$ for vertex $v$ inside sector $S_i$ so that all the $L_i$-directional morphs keep $v$ inside the sector at all times.    
We will separate the proof into two parts.  One part is to show that there exist points $v_i$ in $S_i$ so that the line segment $v_i v_{i+1}$ is parallel to $L_i$.  This is in Section~\ref{sec:nice_pts}.
The other part of the proof is to show that such points $v_i$ ensure that $v$ is inside sector $S$ throughout the morph.  This is in Section~\ref{sec:sectors}. 


For both proofs, we will distinguish the following two possibilities for the relationship between $S_i$ and the line $L_i$ translated to go through point $a$.

\smallskip\noindent{\bf One-sided case.} 
Points $b_i$ and $e_i$ lie in the same closed
half-plane determined by $L_i$.  In this case, whether the sector
$S_i$ is positive or negative, $L_i$ does not intersect the interior
of $S_i$.  See Figure~\ref{fig:one-side}.  An $L_i$-directional morph
keeps $b_i$ and $e_i$ on the same side of $L_i$ so if $S_i$ is
positive it remains positive and if $S_i$ is negative it remains
negative.  Observe that $v$ remains inside the sector if and only if it
    remains inside $D$ and between the two lines $ab$ and
    $ae$.
\begin{figure}[!ht]
  \centering 
  \subfigure[\label{fig:one-side-normal}]{\includegraphics[scale=1.5]{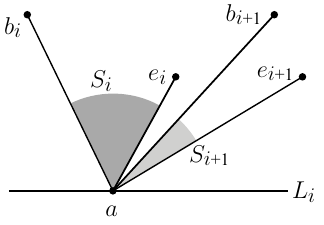}} \hspace{40pt} 
  \subfigure[\label{fig:one-side-sub}]{\includegraphics[scale=2]{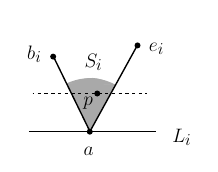}}
  \caption{The one-sided case where $S_i$ lies to one side of $L_i$,
    illustrated for a positive sector $S_i$.  (a) An $L_i$-directional
    morph to $S_{i+1}$.  (b) $v$ remains inside the sector if and only if it
    remains inside $D$ and between the two lines $ab$ and
    $ae$. 
    }
  \label{fig:one-side} 
\end{figure}

\smallskip\noindent{\bf Two-sided case.} 
Points $b_i$ and $e_i$ lie on opposite sides of
$L_i$.  In this case $L_i$ intersects the interior of the sector
$S_i$.  See Figure~\ref{fig:two-side}.  During an $L_i$-directional
morph the sector $S_i$ may remain positive, or it may remain negative,
or it may switch between the two, although it can only switch once.
Observe that $v$ remains
    inside the sector if and only if it remains inside $D$ and on the same side
    of the lines $bb'$ and $ee'$. 
\begin{figure}[!ht]
  \centering 
  \subfigure[\label{fig:inverting}]{\includegraphics[scale=1.7]{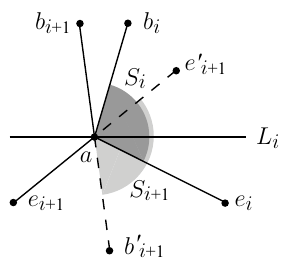}} \hspace{40pt}
  \subfigure[\label{fig:two-side-sub}]{\includegraphics[scale=2]{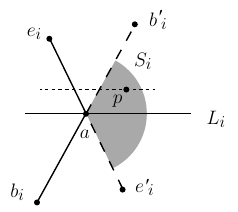}}
  \caption{The two-sided case where $S_i$ contains points on both
    sides of $L_i$.  \subref{fig:inverting} An $L_i$-directional morph from the positive
    sector $S_i$ bounded by $ b_i a e_i$ to the negative sector
    $S_{i+1}$ bounded by $ e'_{i+1} a b'_{i+1}$.  \subref{fig:two-side-sub} $v$ remains
    inside the sector if and only if it remains inside $D$ and on the same side
    of the lines $bb'$ and $ee'$.  
      }
  \label{fig:two-side}
\end{figure}

\remove{
\rednote{Fidel's version}
To prove 
the Theorem 
our aim will be to
place $v$ inside a sector $S$ of a small disk $D$ around $a$
throughout the pseudo-morph $\MM$. The sector $S$ will be delimited by
the lines through $ab$ and $ae$. As the unidirectional morphs defining
$\MM$ act on the sequence of drawings, the sector $S$ will change. We
will show that there exists a set of ``nice points'' which is a subset
of $S$ where $v$ can be placed so that we obtain a true planar
morph. In Section~\ref{sec:sectors} we define the disk $D$ and the
sector $S$ in which we will focus our attention. We also show an
important lemma related to the points remaining inside sector $S$
during a unidirectional morph. In Section~\ref{sec:nice_pts} we
formally define what the set of ``nice points'' is and we show that it
is not empty. We will see that showing that the set of ``nice points''
is non empty implies
Lemma~\ref{lem:penta}. 
Finally, in Section~\ref{sec:runtime} we show that obtaining a true morph from
a pseudo-morph consisting of $O(k)$ steps takes $O(k)$ time.
}

\subsubsection{Sectors and the betweenness property}
\label{sec:sectors}

In this subsection we show that if we can choose points $v_i$ in sector $S_i$ such that the line $v_i v_{i+1}$ is parallel to $L_i$ then $v$ remains inside the sector throughout the morph.

Our main tool is the following lemma proving that ``sidedness'' on line $L$ is preserved in an $L$-directional morph.

\begin{lemma}\label {lem:move_boundary} 
  Let $L$ be a horizontal line and $p_{0},p_{1},q_{0},q_{1}$ be points
  in $L$. Consider a point $p$ that moves at constant speed from
  $p_{0}$ to $p_{1}$ in one unit of time. If $q_{i}$ is to the right
  of $p_{i}$, $i=0,1$, and $q$ is a point that moves at constant speed
  from $q_{0}$ to $q_{1}$ in one unit of time, then $q$ is to the right of $p$ during their entire movement. Note that $p_0$ may lie to the   right or left of $p_1$ and likewise for $q_0$ and $q_1$.
\end{lemma}
\begin{proof}
  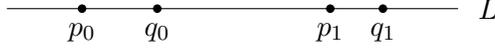
\begin{figure}[!h] \centering 
    \usetikzlibrary{decorations,decorations.markings,decorations.text,calc,arrows}

\begin{tikzpicture} [
  _vertex/.style ={circle,draw=black, fill=black,inner sep=1pt},
__vertex/.style ={circle,draw=black, fill=white,inner sep=1pt},
c_vertex/.style ={circle, double=white, draw=black, fill=black,inner sep=1pt},
gr_vertex/.style ={circle,draw=gray, fill=gray,inner sep=1pt},
dgr_vertex/.style ={circle,draw=black!60, fill=black!60,inner sep=1pt},
  r_vertex/.style={circle,draw=red,  fill=red,inner sep=1pt},
  g_vertex/.style={circle,draw=green,fill=green,inner sep=1pt},
  b_vertex/.style={circle,draw=blue, fill=blue,inner sep=1pt},
  _edge/.style={black,line width=0.5pt},
  gr_edge/.style={gray,line width=0.3pt},
  dgr_edge/.style={black!60,line width=0.3pt},
  r_edge/.style={red,line width=0.3pt},
  g_edge/.style={green,line width=0.3pt},
  b_edge/.style={blue,line width=0.3pt},
 every edge/.style={draw=black,line width=0.3pt}]

\draw (0,0)--(6,0) node[label=east:$L$] {};

\node[_vertex,label=south:$p_0$] (p0) at (1,0) {};
\node[_vertex,label=south:$q_0$] (q0) at (2,0) {};
\node[_vertex,label=south:$p_1$] (p1) at (4.3,0) {};\node[_vertex,label=south:$q_1$] (q1) at (5,0) {};

\end{tikzpicture}

    \caption{Points $p$ and $q$ move from $p_{0}$ to $p_{1}$ and from
      $q_{0}$ to $q_{1}$ respectively.}
    \label{fig:move_boundary}
  \end{figure}
  Let $p_{i}$ and $q_{i}$, $i=0,1$, be points as described above, see
  Figure~\ref{fig:move_boundary}.  Denote by $p_{t}$ and $q_{t}$ the
  positions of $p$ and $q$, for $0\leq t\leq 1$. First note that
  \begin{equation}
    q_{i}=p_{i}+\delta_{i}\label{eq:yi}
  \end{equation}
  for $i=0,1$, with $\delta_{i}>0$. Since $p$ and $q$ are moving at
  constant speed, we have $p_{t}=(1-t)p_{0}+tp_{1}$ and
  $q_{t}=(1-t)q_{0}+tq_{1}$. Now, using equation~\eqref{eq:yi} in the
  expression for $q_{t}$ we have
  \begin{align*}
    q_{t}&=(1-t)(p_{0}+\delta_{0})+t(p_{1}+\delta_{1})\\
&=p_{t}+(1-t)\delta_{0}+t\delta_{1},
  \end{align*}
  where $(1-t)\delta_{0}+t\delta_{1}>0$.
\end{proof}

Clearly, Lemma~\ref{lem:move_boundary} generalizes to a directed line $L$ that is not necessarily horizontal, if we interpret ``to the right of'' as ``further in the direction of'' $L$.


\begin{cor} 
  \label{cor:line-side} Consider an $L$-directional morph acting on
points $p$, $r$ and $s$.  If $p$ is to the right of the line through
$rs$ at the beginning and the end of the $L$-directional morph, then
$p$ is to the right of the line through $rs$ throughout the
$L$-directional morph.
\end{cor}

To prove Corollary~\ref{cor:line-side}, consider a line $L'$ through
$p$ that is parallel to $L$ and apply Lemma~\ref{lem:move_boundary} to
points $p$ and $q$, where $q$ is the intersection of $L'$ and the line
through $r$ and $s$.

We now present our main result about the relative positions of points
$v_i$ and $v_{i+1}$.

\begin{lemma}
  \label{lemma:p-move} If point $v_i$ lies in sector $S_i$ and point
$v_{i+1}$ lies in sector $S_{i+1}$ and the line segment $v_i v_{i+1}$ is
parallel to $L_i$ then an $L_i$-directional morph from $S_i, v_i$ to
$S_{i+1}, v_{i+1}$ keeps the point in the sector at all times.
\end{lemma}

\begin{proof} 
  First consider the one-sided case.  Suppose $S_i$ is a positive
  sector (the case of a negative sector is similar). Observe that $v$ remains in the sector during an $L_i$-directional morph if
  and only if it remains in the disk $D$ and remains between the lines
  $ab$ and $ae$.  See Figure~\ref{fig:one-side-sub}.  Because $v_i$ and
  $v_{i+1}$ both lie in $D$, the line segment between them
  also lies in the disk, and $v$ remains in $D$ throughout the morph.
  In the initial configuration, $v_i$ lies between the lines $ab_i$
  and  $ae_{i}$, and in the final configuration $v_{i+1}$ lies between
  the lines $ab_{i+1} $ and $ae_{i+1}$.  Therefore, by
  Corollary~\ref{cor:line-side}, $v$ remains between the lines
  throughout the $L_i$-directional morph.  Thus $v$ remains inside the
  sector throughout the morph.

  Now consider the two-sided case.  Observe that a point $v$ remains
  in the sector during an $L_i$-directional morph if and only if it
  remains on the same side of the lines $bb'$ and $ee'$.  Note that
  this is true even when the sector changes between positive and
  negative, since the points in $S$ are those in $D$ that are between
  lines $bb'$ and $ee'$.  See Figure cionzo~\ref{fig:two-side-sub}.  As in the
  one-sided case, $v$ remains in the disk throughout the morph.  Also,
  $v$ is on the same side of the lines $b b'$ and $e e'$ in the
  initial and final configurations, and therefore by
  Corollary~\ref{cor:line-side}, $v$ remains on the same side of the
  lines throughout the morph.  Thus $v$ remains inside the sector
  throughout the morph.
\end{proof}

\subsubsection{Nice points}
\label{sec:nice_pts}

In this subsection we prove that there exist points $v_i$ in $S_i$ such that the line segment $v_i v_{i+1}$ is parallel to $L_i$.
We call the possible positions for $v_i$
inside sector $S_i$ the \emph{nice} points, defined formally as
follows:
\begin{itemize}
\item All points in the interior of $S_{k+1}$ are nice.
\item For $1 \leq i\leq k$, a point $v_i$ in the interior of $S_i$
is nice if there is a nice point $v_{i+1}$ in $S_{i+1}$ such that $v_i
v_{i+1}$ is parallel to $L_i$.
\end{itemize}

It suffices to show that all the sets of nice points
are non-empty.  We will in fact characterize the sets. Given a line
$L$, an $L$-\emph{truncation} of a sector $S$ is the intersection of $S$ with
an open slab that is bounded by two lines parallel to $L$ and
contains $a$. Observe that this implies that the open slab will
contain all points of $S$ in a small neighborhood of $a$.  In
particular, an $L$-truncation of a sector is non-empty.
\begin{lemma}
  \label{lemma:nice} The set of nice points in $S_i$ is an
$L_i$-truncation of $S_i$ for $i=1, \ldots, k$.
\end{lemma}
\begin{proof} 
  Let $N_i$ denote the nice points in $S_i$.  The proof is by
  induction as $i$ goes from $k+1$ to 1.  All the points in the interior
  of $S_{k+1}$ are nice.  Suppose by induction that $N_{i+1}$ is an
  $L_{i+1}$-truncation of $S_{i+1}$.

  Consider the one-sided case.  See
  Figure~\ref{fig:transition_case_1_2}.  The slab determining $N_i$
  consists of all lines parallel to $L_i$ that go through a point of
  $N_{i+1}$. 
  The slab is non-empty since $N_{i+1}$ contains all
  of $S_{i+1}$ in a small neighborhood of $a$.  Thus the slab contains
  all points of $S_i$ in a neighborhood of $a$, and thus $N_i$ is an
  $L_{i}$-truncation of $S_i$.

  \begin{figure}[!h] 
    \centering
    \includegraphics{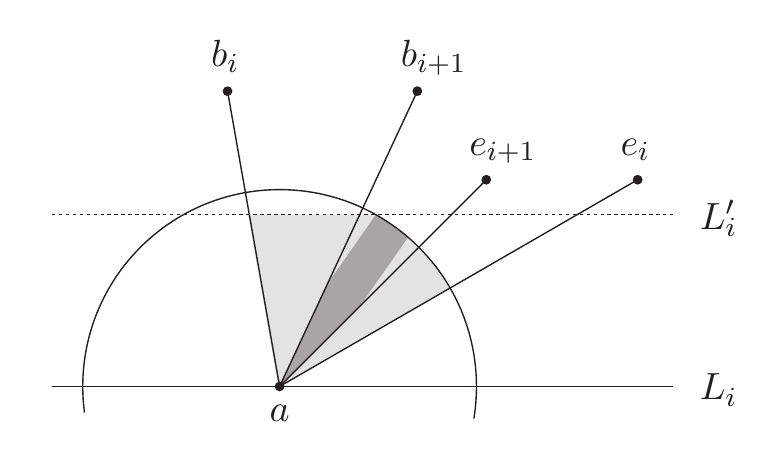}
    \caption{$N_i$ (lightly shaded) is an $L_i$-truncation of $S_i$ in
      the one-sided case.  $N_{i+1}$ is darkly shaded. 
      The slab boundaries for $N_i$ consist of $L'_i$ and a parallel line just below $L_i$. 
    }
    \label{fig:transition_case_1_2}
  \end{figure}
  
  Consider the two-sided case.  See
  Figure~\ref{fig:pie_transition_case_2_1}.  The slab determining
  $N_i$ consists of all lines parallel to $L_i$ that go through a
  point of $N_{i+1}$.  The slab contains $a$ in its interior and thus
  $N_i$ is an $L_{i}$-truncation of $S_i$.
\end{proof}
  \begin{figure}[!ht] 
    \centering 
    \subfigure[]{\includegraphics[scale=.9]{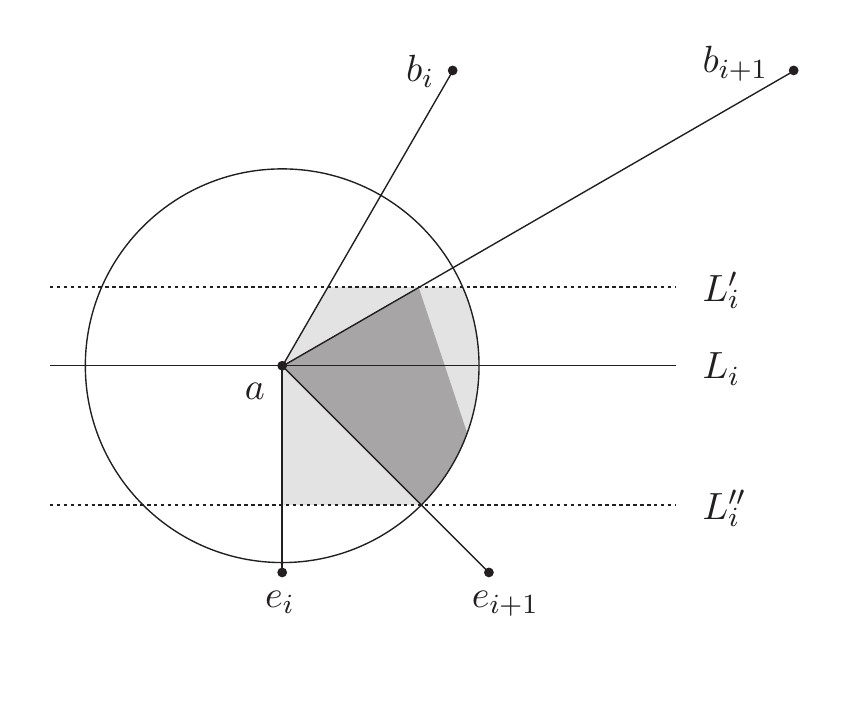}}
    \hspace{40pt} 
    \subfigure[]{\includegraphics[scale=.9]{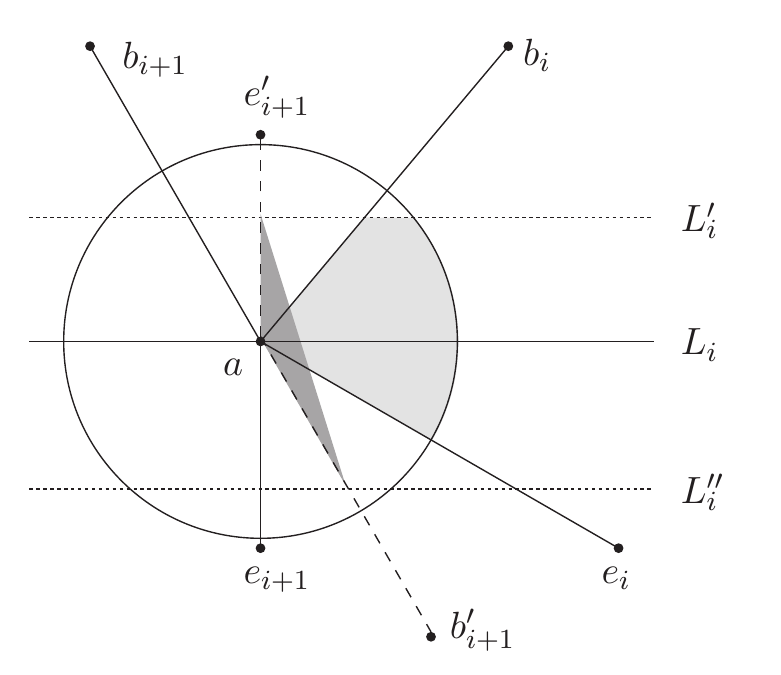}}
    \caption{$N_i$ (lightly shaded) is an $L_i$-truncation of $S_i$ in
      the two-sided case.  $N_{i+1}$ is darkly shaded.  $L'_i$ and
      $L''_i$ are the slab boundaries for $N_i$. 
      }
    \label{fig:pie_transition_case_2_1}
  \end{figure}

Lemma~\ref{lemma:nice} implies in particular that the set of nice
points is non-empty.

The last thing we need to do to complete the proof of Theorem~\ref{thm:pseudomorph} for the case where $v$ has degree 5, is to show that the above procedure to compute the $v_i$'s has a running time of $O(k)$.
It suffices to show how to compute $N_i$ from $N_{i+1}$ in constant time. 
To do this, we just compute the maximum and minimum points of $S_{i+1}$ in the direction perpendicular to $L_i$.  The slab boundaries for $N_i$ go through these points.  Then $N_i$ is the intersection of the slab with sector $S_i$.

\remove{
\rednote{runtime}
Now it suffices to show that it takes constant time to compute the set
of nice points $N_{i}$ in $\Gamma_{i}$, for $0\leq i\leq k$. The
argument we provide is inductive as $i$ goes from $k$ to $0$. Note
that $N_{k}=S_{k}$ which is clearly computable in constant time. Now,
given $N_{i+1}$ we may obtain the points in $N_{i}$ in constant time
since $N_{i}$ is an $L_{i}$-truncation of $S_{i}$. This truncation can
be computed in constant time since we just require to obtain the
intersection of $S_{i}$ with two half-planes defined by lines parallel
to $L_{i}$. Hence we require $O(k)$ time to obtain the sets of nice
points. Therefore the overall time required to obtain a position for
$v$ in the intermediate drawings of $\MM$ is $O(k)$ by just choosing a
point in $N_{i+1}$ and have it follow the $L_{i}$-directional morph to
reach $N_{i}$, for $1\leq i\leq k$. This provides the last ingredient
in the proof of Lemma~\ref{lem:penta}.
}


\remove{
\subsubsection{Running time}\label{sec:runtime}

\rednote{Fix $k$ to $k+1$.}
In this section we obtain an upper bound on the running time of the
algorithm that results from Lemma~\ref{lem:penta}. Suppose
$\MM=\langle \Gamma_{0},\Gamma_{1},\ldots, \Gamma_{k} \rangle$ is a
planar pseudo-morph that contracts vertex $v$ and consists of
unidirectional steps. Let us denote by $M$ the planar morph obtained
from $\MM$. First we discuss how to obtain the radius of the disk $D$
defining sector $S$ in $O(k)$ time. Then we describe how the set of
nice points can be obtained from $S$ throughout the pseudomorph in
$O(k)$ time.

Note that as a first step we must define the disk $D$ centered at $a$
of radius $\varepsilon$ that defines the sector $S$. To achieve this
it suffices to choose $\varepsilon>0$ small enough so that the set of
points in $D$ that can see all the vertices of $P$ defines a
positive-area sector $S$ of $D$. Let
$\varepsilon_{i}=\min_{x\in\{b,c,d,e\}}\{\Vert a-x \Vert\}/2$ during
the morph $\langle\Gamma_{i-1},\Gamma_{i} \rangle$, for
$1\leq i\leq k$. Observe that for each pair of consecutive
intermediate drawings $\Gamma_{i-1}$ and $\Gamma_{i}$,
$\varepsilon_{i}$ can be computed in constant time since we are just
computing the minimum distance from $a$ to the $4$ line segments
$b_{i-1}b_{i}$, $c_{i-1}c_{i}$, $d_{i-1}d_{i}$ and
$e_{i-1}e_{i}$. Finally note that it suffices to take
$\varepsilon=\min_{1\leq i\leq s}\{\varepsilon_{i}\}$ to define
$D$. Thus requiring $O(k)$ time.

Now it suffices to show that it takes constant time to compute the set
of nice points $N_{i}$ in $\Gamma_{i}$, for $0\leq i\leq k$. The
argument we provide is inductive as $i$ goes from $k$ to $0$. Note
that $N_{k}=S_{k}$ which is clearly computable in constant time. Now,
given $N_{i+1}$ we may obtain the points in $N_{i}$ in constant time
since $N_{i}$ is an $L_{i}$-truncation of $S_{i}$. This truncation can
be computed in constant time since we just require to obtain the
intersection of $S_{i}$ with two half-planes defined by lines parallel
to $L_{i}$. Hence we require $O(k)$ time to obtain the sets of nice
points. Therefore the overall time required to obtain a position for
$v$ in the intermediate drawings of $\MM$ is $O(k)$ by just choosing a
point in $N_{i+1}$ and have it follow the $L_{i}$-directional morph to
reach $N_{i}$, for $1\leq i\leq k$. This provides the last ingredient
in the proof of Lemma~\ref{lem:penta}.
}

\remove{
\subsection{The exceptional cases}
\label{sec:exceptional_cases}

In this section we show how to obtain a true planar morph from a
planar pseudo-morph $\MM$ where vertex $v$ of degree at most 2 is
contracted and consists only of unidirectional steps. Before stating
this formally, let us prove a technical lemma which will be useful.

\begin{lemma}\label{lemma:dummy_zs}
  Let $\mm=\langle\Gamma_{1},\ldots,\Gamma_{k}\rangle$ be a planar
  morph of $G$ consisting of unidirectional steps. Suppose
  $\langle\Gamma_{i},\Gamma_{i+1}\rangle$ is an $L_{i}$-directional
  morph for $1\leq i\leq k-1$. If $ab$ is an edge of $G$, then there
  exist points $c_{1},\ldots,c_{k}$ and $\alpha>0$ such that:
  \begin{itemize}
    \item $c_{i}c_{i+1}$ is parallel to $L_{i}$ for $1\leq i\leq k-1$, 
    \item the size of the clockwise angle $c_{i}a_{i}b_{i}$ is $\alpha$ for
      $1\leq i\leq k$
    \item apart from $ab$ no other edge incident to $a$ is in the
      sector determined by $cab$ throughout $\mm$.
  \end{itemize}
\end{lemma}
\begin{proof}
  We begin by defining $\alpha$. Let $\alpha_{i}$ and $\alpha_{i}'$ be
  the angle formed by $a_{i}b_{i}$ and $L_{i}$ and $a_{i+1}b_{i+1}$
  and $L_{i}$ respectively, $1\leq i \leq k-1$. If $\Deg(a)\geq 2$ let
  $d$ be the neighbor of $a$ that follows $b$ in clockwise order and
  let $\beta_{0}$ be the minimum size of the angle $dab$ during $\mm$,
  otherwise let $\beta_{0}$ be $\pi$. Among the angles $\alpha_{i}$
  and $\alpha_{i}'$ that are positive, $1\leq i\leq k-1$, let
  $\alpha_{0}$ and $\alpha_{0}'$ denote the respective minima. We take
  $\alpha$ to be the half the minimum of $\alpha_{0}$, $\alpha_{0}'$
  and $\beta_{0}$, that is
  $\alpha=1/2\min\{\alpha_{0},\alpha_{0}',\beta_{0}\}$.

  We now define the points $c_{1},\ldots,c_{k}$. We will assume that
  $a_{i}$ is fixed by the $L_{i}$-directional morph. Let $c_{1}$ be
  any point so that $c_{1}a_{1}b_{1}$ define a clockwise angle of size
  $\alpha$. Suppose we have defined the point $c_{i}$, let us show how
  to define point $c_{i+1}$. If $L_{i}$ is parallel to $a_{i}b_{i}$
  then we let $c_{i+1}$ be $c_{i}$. Clearly $c_{i}c_{i+1}$ is parallel
  to $L_{i}$ and the angle $c_{i+1}a_{i+1}b_{i+1}$ remains of size
  $\alpha$. Now, suppose $L_{i}$ is not parallel to
  $a_{i}b_{i}$. Since $\alpha<\alpha_{i}$ it follows that both $b_{i}$
  and $c_{i}$ are on the same half-plane of $L_{i}$. Let $L_{i}'$ be
  the line parallel to $L_{i}$ through $c_{i}$ and let $L$ be the line
  through $a$ forming an angle of $\alpha$ with $a_{i+1}b_{i+1}$. The
  fact that $\alpha<\alpha_{i}'$ guarantees that the intersection of
  $L$ and $L_{i}'$ remains on the same side of the segment
  $ab$. Observe that $L$ and $L_{i}'$ are not parallel. The point of
  intersection of $L$ and $L_{i}'$ is in the same half-plane of
  $L_{i}$ as $c_{i}$ since $L_{i}'$ is contained in that
  half-plane. See Figure~\ref{fig:new_dummy}. We define $c_{i+1}$ to
  be the point of intersection of $L_{i}'$ and $L$. Now we clearly
  have that $c_{i}c_{i+1}$ is parallel to $L_{i}$ and that the size of
  the angle $c_{i+1}a_{i+1}b_{i+1}$ is $\alpha$.
  \begin{figure}[h]
    \centering

\begin{tikzpicture}[
  _vertex/.style ={circle,draw=black, fill=black,inner sep=1pt},
__vertex/.style ={circle,draw=black, fill=white,inner sep=1pt},
c_vertex/.style ={circle, double=white, draw=black, fill=black,inner sep=1pt},
gr_vertex/.style ={circle,draw=gray, fill=gray,inner sep=1pt},
  r_vertex/.style={circle,draw=red,  fill=red,inner sep=1pt},
  g_vertex/.style={circle,draw=green,fill=green,inner sep=1pt},
  b_vertex/.style={circle,draw=blue, fill=blue,inner sep=1pt},
  _edge/.style={black,line width=0.5pt},
  gr_edge/.style={gray,line width=0.3pt},
  r_edge/.style={red,line width=0.3pt},
  g_edge/.style={green,line width=0.3pt},
  b_edge/.style={blue,line width=0.3pt},
 every edge/.style={draw=black,line width=0.3pt}]

 \begin{scope}
   \coordinate (x) at (0,0); 
   \coordinate (yi) at (150:20mm);
   \coordinate (yii) at (45:14.14mm);
   \coordinate (zi) at (140:12mm);

	\coordinate(zii) at(35:13.44mm);

	\fill[black!60] (x)--(140:8mm) arc [start angle=140,end angle=150,radius=8mm]--(x);

	\fill[black!60] (x)--(35:8mm) arc [start angle=35,end angle=45,radius=8mm]--(x);

	\draw[gr_edge,dashed] ($-1*(zii)$) -- ($2*(zii)$) node[label={[black]east:$L$}] {};
	\draw[gr_edge,dashed] ($(zi) - (15mm,0)$) -- ($(zi) + (35mm,0)$) node[label={[black]east:$L_i'$}] {};
	\draw (-28mm,0) -- (30mm,0) node[label={[black]east:$L_i$}] {};

   \draw[_edge] (x)--(yi);
   \draw[_edge] (x)--(zi);
   \draw[gr_edge,dashed] (x)--(yii);

   \node[_vertex,label=south:$a$] (X) at (x) {};
   \node[__vertex,label=north:$b_i$] (Y) at (yi) {};

   \node[__vertex,label=north east:$c_i$] (Z) at (zi) {};

   \node[gr_vertex,label={[gray]north:$b_{i+1}$}] (Y) at (yii) {};
   \node[gr_vertex,label={[gray]south east:$c_{i+1}$}] (Z) at (zii) {};

   \node (alabel) at (270:15mm) {(a)};
 \end{scope}

\def\brad{10}
\def\bsang{185}
\def\bfang{115}

 \begin{scope}[xshift=75mm]

   \coordinate (x) at (0,0); 
   \coordinate (yi) at (150:20mm);
   \coordinate (yii) at (45:14.14mm);
   \coordinate (zi) at (140:12mm);

	\coordinate(zii) at(35:13.44mm);

	\fill[black!60] (x)--(35:8mm) arc [start angle=35,end angle=45,radius=8mm]--(x);

	\draw[gr_edge,dashed] ($-1*(zii)$) -- ($2*(zii)$) node[label={[black]east:$L$}] {};
	\draw[gr_edge,dashed] ($(zi) - (15mm,0)$) -- ($(zi) + (35mm,0)$) node[label={[black]east:$L_i'$}] {};
	\draw (-28mm,0) -- (30mm,0) node[label={[black]east:$L_i$}] {};

	\draw[gr_edge,dashed] (yi) -- (yii);
   \draw[_edge] (x)--(yii);
   \draw[_edge] (x)--(zii);

   \node[_vertex,label=south:$a$] (X) at (x) {};
   \node[gr_vertex,label={[gray]north:$b_i$}] (Y) at (yi) {};
   \node[gr_vertex,label={[gray]north east:$c_i$}] (Z) at (zi) {};

   \node[__vertex,label=north:$b_{i+1}$] (Y) at (yii) {};

   \node[__vertex,label=south east:$c_{i+1}$] (Z) at (zii) {};

   \node (alabel) at (270:15mm) {(b)};
 \end{scope}

\end{tikzpicture}
    \caption{We choose $c_{i+1}$ to be the intersection of $L$ and $L_{i}'$.}
    \label{fig:new_dummy}
  \end{figure}

  To conclude the proof, note that the sector defined by $cab$ remains
  free of other edges incident to $a$ since $\mm$ is planar and the
  angle $\alpha$ is smaller than $\beta_{0}$.
\end{proof}

We now state the main result of this section. For the proof we will
consider several cases, depending on the vertices involved in the
contraction. Lemma~\ref{lemma:dummy_zs} will be useful in two of these
cases.

\begin{lemma}\label{lem:low_degree}
  Let $\MM$ be a planar pseudo-morph from $\Gamma$ to $\Gamma'$ that
  contracts vertex $v$ and consists of unidirectional steps. If
  $\Deg(v)\in\{1,2\}$, then we can obtain a planar morph $M$ from
  $\Gamma$ to $\Gamma'$ consisting of the same number of
  unidirectional steps.
\end{lemma}
\begin{proof}
  Let us consider different cases, depending on the degree of vertex
  $v$ and the degree of the vertex $a$ onto which $v$ is
  contracted. We use $\mm=\langle\Gamma_{1},\ldots,\Gamma_{k}\rangle$
  to denote the morph of the reduced graph in $\MM$. 



\begin{enumerate}
  \item \textbf{[$\deg(v)=1$]} We consider three further subcases,
    depending on the degree of $a$.
    
    \begin{enumerate}

    \item \textbf{[$\deg(a)=1$]} Since the graph is connected, the
      pseudo-morph $\mm$ amounts to a single translation of vertex
      $a$. Thus the morph $\MM$ consists of $3$ steps. It can be seen
      that we can obtain a true morph in $3$ steps that satisfies the
      required conditions.
    
  \item\label{item:deg1-deg2-case} \textbf{[$\deg(a)=2$]} Let $b$ be
    the second neighbor of $a$. In this case we place $v$ close
    enough to $a$ along the extension of the segment $ab$. Formally
    speaking, let $N$ denote the maximum distance between $a$ and $b$
    during $\mm$ and choose $\delta$ so that the disk centered at $a$
    and radius $\delta$ only contains $a$ and its only incident edge
    during $\mm$. Then we may place $v$ at
    $a+\delta/(2N)(a-b)$. The fact that $v$ will follow $a$ and $b$
    unidirectionally follows from an argument similar to the one used
    to prove Lemma~\ref{lemma:convex-comb}.

  \item \textbf{[$\deg(a)\geq 3$]} Let $b$ be the neighbor of $a$
    that precedes $v$ in clockwise order around $a$. Similarly, let
    $e$ be the neighbor of $a$ that succeeds $v$ in clockwise order
    around $x$. The idea to handle this case will be similar to the
    one followed to handle vertices of degree 5 in
    Section~\ref{sec:penta}. We will consider a small disk $D$ around
    $a$ so that we  place $v$ in the sector of $D$ between the lines
    through $ab$ and $ae$. See Figure~\ref{fig:deg1-deg3-case}.

    Let $\varepsilon>0$ be small enough so that the disk $D$ centered
    at $a$ only contains $a$ and edges incident to $a$ throughout
    $\mm$. As in Section~\ref{sec:sectors} we let $S$ be the positive
    sector in $D$ defined by $ab$ and $ae$ whenever the angle $eab$ is
    convex; and let $S$ be the negative sector in $D$ defined by the
    extensions of $ab$ and $ae$ whenever the angle $eab$ is
    reflex. Note that we may place $v$ at any nice point inside $S$
    throughout $\mm$ to obtain a true morph of $G$. Thus by
    Lemma~\ref{lemma:nice} it follows that such planar morph from
    $\Gamma$ to $\Gamma'$ that has the same number of steps as $\MM$
    exists.
    \begin{figure}[h]
      \centering
\begin{tikzpicture}[
  _vertex/.style ={circle,draw=black, fill=black,inner sep=1pt},
__vertex/.style ={circle,draw=black, fill=white,inner sep=1pt},
c_vertex/.style ={circle, double=white, draw=black, fill=black,inner sep=1pt},
gr_vertex/.style ={circle,draw=gray, fill=gray,inner sep=1pt},
  r_vertex/.style={circle,draw=red,  fill=red,inner sep=1pt},
  g_vertex/.style={circle,draw=green,fill=green,inner sep=1pt},
  b_vertex/.style={circle,draw=blue, fill=blue,inner sep=1pt},
  _edge/.style={black,line width=0.5pt},
  gr_edge/.style={gray,line width=0.3pt},
  r_edge/.style={red,line width=0.3pt},
  g_edge/.style={green,line width=0.3pt},
  b_edge/.style={blue,line width=0.3pt},
 every edge/.style={draw=black,line width=0.3pt}]

 \begin{scope}

   \coordinate (x) at (0,0); 
   \coordinate (w) at (170:13mm);
   \coordinate (v) at (0:10mm);
   \coordinate (y) at (45:13mm);

   \draw[_edge] (x)--(w); 
   \draw[_edge] (x)--(y);
   \draw[_edge] (x)--(v);

   \node[_vertex,label=north:$a$] (X) at (x) {};
   \node[__vertex,label=south:$e$] (W) at (w) {};
   \node[__vertex,label=west:$b$] (y) at (y) {};
   \node[__vertex,label=east:$v$] (V) at (v) {};

   \node (alabel) at (270:15mm) {(a)};
 \end{scope}

 \begin{scope}[xshift=4cm]

   \coordinate (x) at (0,0); 
   \coordinate (w) at (170:13mm);
   \coordinate (vp) at (0:10mm);
   \coordinate (v) at (270:3mm);
   \coordinate (y) at (45:13mm);

   \coordinate (my) at ($-1*(w)$);
   \coordinate (mz) at ($-1*(y)$);

   \fill[gray!20] (0,0) -- ($0.5*(mz)$) arc [start angle=225, end angle=350, radius=6.5mm];

   \draw[gray] (mz) -- ($1.2*(y)$);
   \draw[gray] (my) -- ($1.2*(w)$);

   \draw[_edge] (x)--(w); 
   \draw[_edge] (x)--(y);
   \draw[_edge] (x)--(v);
    \draw[gr_edge,dashed] (v)--(vp);   

   \node[_vertex,label=north:$a$] (X) at (x) {};
   \node[__vertex,label=south:$e$] (W) at (w) {};
   \node[__vertex,label=west:$b$] (Y) at (y) {};
   \node[__vertex,label=south:$v$] (V) at (v) {};
   \node[gr_vertex,label=east:$\color{gray}v$] (V) at (vp) {};

   \node (alabel) at (270:15mm) {(b)};
 \end{scope}

\end{tikzpicture}
      \caption{The case where $\deg(v)=1$ and $\deg(a)\geq 3$.}
      \label{fig:deg1-deg3-case}
    \end{figure}

  \end{enumerate}

\item \textbf{$\deg(v)=2$} Let $b$ denote the second neighbor of
  $v$. Let us consider two possible subcases.
  \begin{enumerate}
  \item \textbf{[$ab\not\in E(G)$]} Here we let $v$ lie at a fixed
    convex combination of $a$ and $b$ throughout the morph, say
    $v=(1/2)a+(1/2)b$. See Figure~\ref{fig:deg2-no-edge}. It follows
    from Lemma~\ref{lemma:convex-comb} that $v$ will move
    unidirectionally since $\mm$ consists of unidirectional steps.
    \begin{figure}[h]
      \centering
      \begin{tikzpicture}[
  _vertex/.style ={circle,draw=black, fill=black,inner sep=1pt},
__vertex/.style ={circle,draw=black, fill=white,inner sep=1pt},
c_vertex/.style ={circle, double=white, draw=black, fill=black,inner sep=1pt},
gr_vertex/.style ={circle,draw=gray, fill=gray,inner sep=1pt},
  r_vertex/.style={circle,draw=red,  fill=red,inner sep=1pt},
  g_vertex/.style={circle,draw=green,fill=green,inner sep=1pt},
  b_vertex/.style={circle,draw=blue, fill=blue,inner sep=1pt},
  _edge/.style={black,line width=0.5pt},
  gr_edge/.style={gray,line width=0.3pt},
  r_edge/.style={red,line width=0.3pt},
  g_edge/.style={green,line width=0.3pt},
  b_edge/.style={blue,line width=0.3pt},
 every edge/.style={draw=black,line width=0.3pt}]

 \begin{scope}
   \coordinate (x) at (190:10mm); 
   \coordinate (y) at (10:10mm);
   \coordinate (v) at (110:8mm);

   \draw[_edge] (x)--(v); 
   \draw[_edge] (v)--(y);

   \node[_vertex,label=south:$a$] (X) at (x) {};
   \node[__vertex,label=south:$b$] (Y) at (y) {};
   \node[__vertex,label=south:$v$] (V) at (v) {};

   \node (alabel) at (270:15mm) {(a)};
 \end{scope}

 \begin{scope}[xshift=4cm]

   \coordinate (x) at (190:10mm); 
   \coordinate (y) at (10:10mm);
   \coordinate (v) at ($0.5*(x)+0.5*(y)$);
   \coordinate (vp) at (110:8mm);

   \draw[_edge] (x)--(v); 
   \draw[_edge] (v)--(y);
   \draw[gr_edge,dashed] (v)--(vp);   

   \node[_vertex,label=south:$a$] (X) at (x) {};
   \node[__vertex,label=south:$b$] (Y) at (y) {};
   \node[__vertex,label=south:$v$] (V) at (v) {};
   \node[gr_vertex,label=east:$\color{gray}v$] (V) at (vp) {};

   \node (alabel) at (270:15mm) {(b)};
 \end{scope}

\end{tikzpicture}
      \caption{The case where $\deg(v)=2$ and $ab\not\in E(G)$.}
      \label{fig:deg2-no-edge}
    \end{figure}

  \item \textbf{[$ab\in E(G)$]} Our aim in this case will be to place
    $v$ close to $a$ along a line segment which forms a constant small
    angle with edge $ab$. Without loss of generality we may assume
    that $b$ precedes $v$ in clockwise order around $a$. Let $D$ be a
    disk centered at $a$ with radius $\varepsilon>0$. We choose
    $\varepsilon$ as follows. Take $\varepsilon>0$ small enough so
    that
    \begin{itemize}
    \item during $\mm$, $D$ only contains $a$ and the edges incident to
      $a$, excluding the other endpoints, and
    \item for all neighbors $p$ of $b$ other than $a$, the ray
      $\vec{bp}$ does not intersect $D$.
    \end{itemize}

    By Lemma~\ref{lemma:dummy_zs} there exist points
    $c_{1},\ldots,c_{k}$ and $\alpha>0$ so that    
    \begin{itemize}
    \item $c_{i}c_{i+1}$ is parallel to $L_{i}$, where $L_{i}$ denotes
      the direction of the unidirectional morph
      $\langle\Gamma_{i},\Gamma_{i+1}\rangle$,
    \item the angle $c_{i}a_{i}b_{i}$ has size $\alpha$, and
    \item no other edge incident to $a$ is in the clockwise sector of
      $D$ from $ac$ to $ab$ throughout $\mm$.
    \end{itemize}

    We may observe that in particular the existence of such points
    imply that the $i$-th unidirectional morph in $\mm$ acts
    unidirectionally on the point $c_{i}$, thus we may speak of a
    point $c$ moving throughout $\mm$. It only remains to place $v$
    close enough to $a$ so that $av$ remains inside $D$. We choose to
    place $v$ at a fixed convex combination of $ac$. Let $N$ denote
    the maximum distance between $a$ and $c$. We place $v$ at
    $a+\varepsilon/(2N)(c-a)$. See Figure~\ref{fig:deg2-edge}. Thus the degenerate version of
    Lemma~\ref{lemma:convex-comb} where the triangle collapses to a
    line segment implies the result.

    \begin{figure}[h]
      \centering

\begin{tikzpicture}[
  _vertex/.style ={circle,draw=black, fill=black,inner sep=1pt},
__vertex/.style ={circle,draw=black, fill=white,inner sep=1pt},
c_vertex/.style ={circle, double=white, draw=black, fill=black,inner sep=1pt},
gr_vertex/.style ={circle,draw=gray, fill=gray,inner sep=1pt},
  r_vertex/.style={circle,draw=red,  fill=red,inner sep=1pt},
  g_vertex/.style={circle,draw=green,fill=green,inner sep=1pt},
  b_vertex/.style={circle,draw=blue, fill=blue,inner sep=1pt},
  _edge/.style={black,line width=0.5pt},
  gr_edge/.style={gray,line width=0.3pt},
  r_edge/.style={red,line width=0.3pt},
  g_edge/.style={green,line width=0.3pt},
  b_edge/.style={blue,line width=0.3pt},
 every edge/.style={draw=black,line width=0.3pt}]

 \begin{scope}
   \coordinate (x) at (0,0); 
   \coordinate (y) at (170:18mm);
   \coordinate (v) at (0:10mm);

   \draw[_edge] (x)--(y);
   \draw[_edge] (x)--(v);
   \draw[_edge] (v)--(y);

   \node[_vertex,label=south west:$a$] (X) at (x) {};
   \node[__vertex,label=south:$b$] (Y) at (y) {};
   \node[__vertex,label=south east:$v$] (V) at (v) {};

   \node (alabel) at (270:15mm) {(a)};
 \end{scope}

\def\brad{10}
\def\bsang{185}
\def\bfang{115}

 \begin{scope}[xshift=4cm]
   \coordinate (x) at (0,0); 
   \coordinate (y) at (170:18mm);
   \coordinate (v) at (70:3mm);
   \coordinate (vp) at (0:10mm);

   \coordinate (li) at (70:10mm);

   \draw[black,thick,left to reversed-right to
   reversed,postaction={decorate,decoration={text along path,
       raise=9pt,text={ {$\alpha$}{} },text color=black, text
       align={align=center}}}] (170:11mm) arc [start angle=170, end angle=70, radius=11mm]; 

   \fill[gray!20] (0,0) --(li) arc [start angle=70, end angle=170, radius=10mm] --(0,0); 

   \draw[gray] ($-0.5*(y)$)--($1.5*(y)$);

   \draw[_edge] (x)--(y); 
   \draw[_edge] (x)--(v);
   \draw[_edge] (v)--(y);
   \draw[gr_edge,dashed] (v)--(vp);

   \draw[gray] ($-0.6*(li)$)--($2*(li)$) node[shape=circle,fill=black,inner sep=1pt,label={east:\textcolor{black}{$c$}}] {};

   \node[_vertex,label=south west:$a$] (X) at (x) {};
   \node[__vertex,label=south:$b$] (Y) at (y) {};
   \node[__vertex,label=north:$v$] (V) at (v) {};
   \node[gr_vertex,label=south east:$\color{gray}v$] (V) at (vp) {};

   \begin{scope}[scale=2]

   \end{scope}

   \node (alabel) at (270:15mm) {(b)};
 \end{scope}

\end{tikzpicture}
      \caption{The case where $\deg(v)=2$ and $ab\in E(G)$.}
      \label{fig:deg2-edge}
    \end{figure}

  \end{enumerate}
  \end{enumerate}
  
\end{proof}
}

\section{A Linear Lower Bound}\label{se:lowerbound}

In this section we prove Theorem~\ref{th:lb-bound}, that 
there exist two straight-line planar drawings of an $n$-vertex path such that any planar morph between them consists of $\Omega(n)$ linear morphs.

Specifically we will describe two straight-line planar drawings $\Gamma$ and $\Lambda$ of an $n$-vertex path $P=(v_1,\dots,v_n)$, and prove that any straight-line planarity preserving morph $M$ between $\Gamma$ and $\Lambda$ requires $\Omega(n)$ morphing steps. In order to simplify the description, we consider each edge $e_i=(v_i,v_{i+1})$ as oriented from $v_i$ to $v_{i+1}$, for $i=1,\dots,n-1$.

Drawing $\Gamma$ (see Figure~\ref{fig:lb-path}) is such that all the vertices of $P$ lie on a horizontal straight line with $v_i$ to the left of $v_{i+1}$, for each $i = 1, \dots, n-1$.

Drawing $\Lambda$ (see Figure~\ref{fig:lb-spiral}) is such that: 

\begin{itemize}
\item for each $i = 1, \dots, n-1$ with $i \textrm{ mod } 3 \equiv 1$, $e_i$ is horizontal with $v_i$ to the left of $v_{i+1}$; 
\item for each $i = 1, \dots, n-1$ with $i \textrm{ mod } 3 \equiv 2$, $e_i$ is parallel to line $y=\tan(\frac{2\pi}{3})x$ with $v_i$ to the right of $v_{i+1}$; and 
\item for each $i = 1, \dots, n-1$ with $i \textrm{ mod } 3 \equiv 0$, $e_i$ is parallel to line $y=\tan(-\frac{2\pi}{3})x$ with $v_i$ to the right of $v_{i+1}$.
\end{itemize}

\begin{figure}[htb]
\centering
\hfill
\subfigure[]{\includegraphics[width= 0.4 \textwidth]{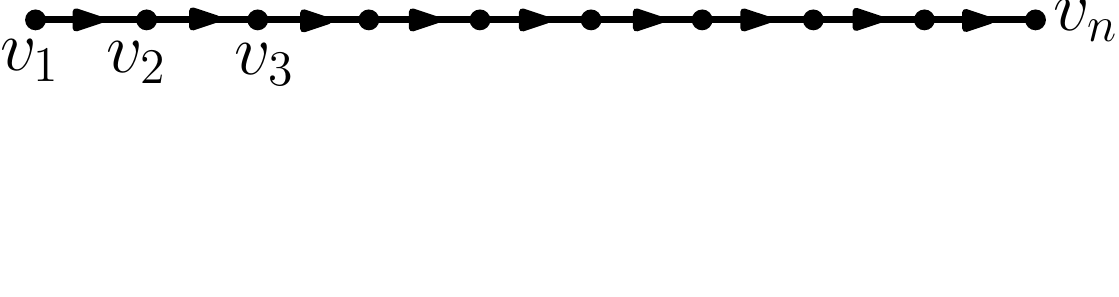}\label{fig:lb-path}}
\hfill
\subfigure[]{\includegraphics[width= 0.3 \textwidth]{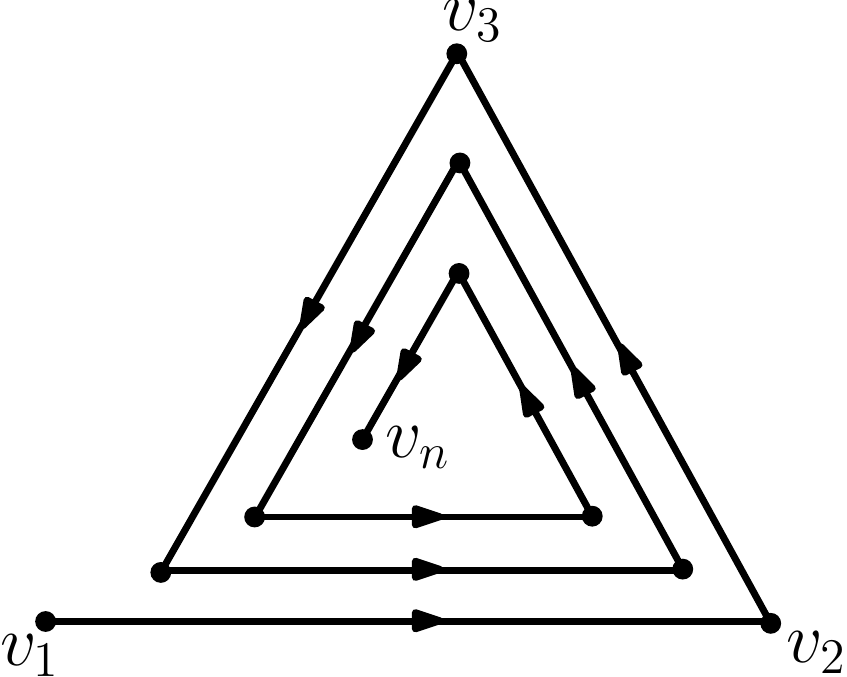}\label{fig:lb-spiral}}
\hfill
\caption{Drawings $\Gamma$ (a) and $\Lambda$ (b).}
\label{fig:lb-drawings}
\end{figure}

Let $M=\morph{\Gamma=\Gamma_1,\dots,\Gamma_x=\Lambda}$ be any planar morph transforming $\Gamma$ into $\Lambda$.

For $i=1,\dots,n$ and $j=1,\dots,x$, we denote by $v_i^j$ the point where vertex $v_i$ is placed in $\Gamma_j$ and by $e_i^j$ the directed straight-line segment representing edge $e_i$ in $\Gamma_j$.

For $1\leq j\leq x-1$, we define the \emph{rotation} $\rho_i^j$ of $e_i$ around $v_i$ during the morphing step $\langle \Gamma_{j},\Gamma_{j+1} \rangle$ as follows (see Figure~\ref{fig:lb-rotation}). Translate $e_i$ at any time instant of $\langle \Gamma_{j},\Gamma_{j+1} \rangle$ so that $v_i$ stays still during the entire morphing step. After this translation, the morph between $e_i^{j}$ and $e_i^{j+1}$ is a rotation of $e_i$ around $v_i$ (where $e_i$ might vary its length during $\langle \Gamma_{j},\Gamma_{j+1} \rangle$) spanning an angle $\rho_i^j$. We assume $\rho_i^j>0$ if the rotation is counter-clockwise and $\rho_i^j<0$ otherwise. We have the following.


\begin{figure}[htb]
\centering
\hfill
\subfigure[\label{fig:lb-rot-a}]{\includegraphics[width= 0.5 \textwidth]{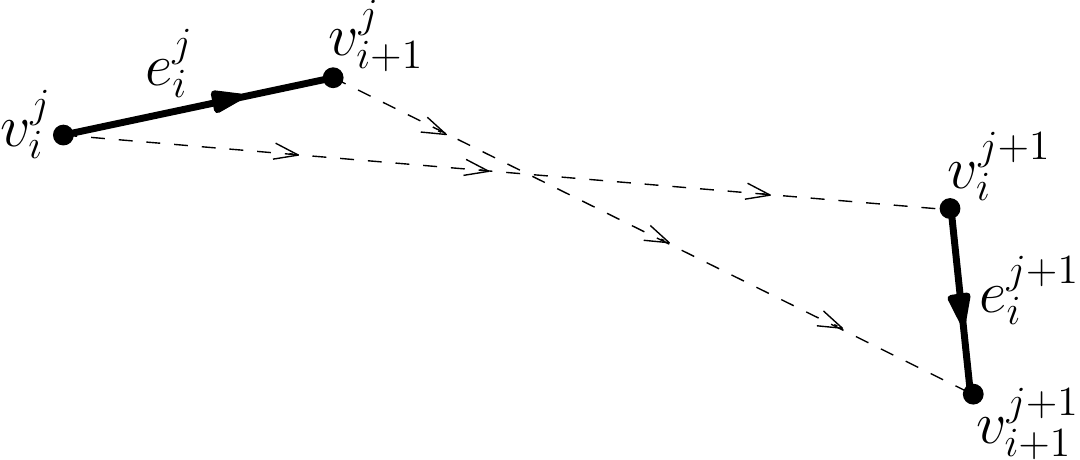}\label{fig:lb-rotation-a}}
\hfill
\subfigure[\label{fig:lb-rot-b}]{\includegraphics[width= 0.27 \textwidth]{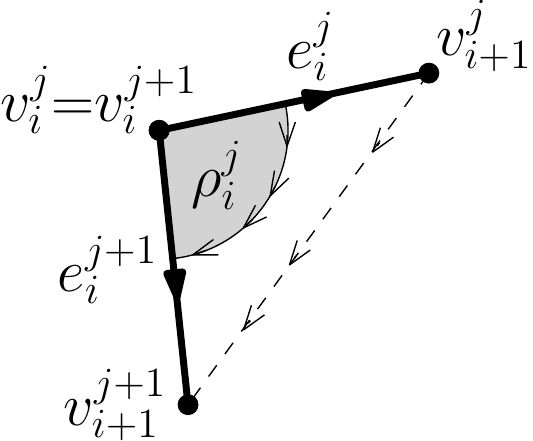}\label{fig:lb-rotation-c}}
\hfill
\caption{\subref{fig:lb-rot-a} Morph between $e_i^{j}$ and $e_i^{j+1}$. \subref{fig:lb-rot-b} Translation of the positions of $e_i$ during $\langle \Gamma_{j},\Gamma_{j+1} \rangle$, resulting in $e_i$ spanning an angle $\rho_i^j$ around $v_i$.}
\label{fig:lb-rotation}
\end{figure}

\begin{lemma}\label{le:lb-rotation}
For each $j=1,\dots,x-1$ and $i=1,\dots,n-1$, we have $|\rho_i^j| < \pi$.
\end{lemma}
\begin{proof}
Assume, for a contradiction, that $|\rho_i^j| \geq \pi$, for some $1\leq j\leq x-1$ and $1\leq i\leq n-1$. Also assume, w.l.o.g., that the morphing step $\langle \Gamma_{j},\Gamma_{j+1} \rangle$ happens between time instants $t=0$ and $t=1$. 

We introduce some notation. For any $0\leq t\leq 1$, we denote: 

\begin{itemize}
\item by $v_i(t)$ the position of $v_i$ at time instant $t$ -- note that $v_i(0)=v_i^j$ and $v_i(1)=v_i^{j+1}$;
\item by $v_{i+1}(t)$ the position of $v_{i+1}$ at time instant $t$ -- note that $v_{i+1}(0)=v_{i+1}^j$ and $v_{i+1}(1)=v_{i+1}^{j+1}$;
\item by $e_i(t)$ the drawing of $e_i$ at time instant $t$ -- note that $e_i(0)=e_i^j$ and $e_i(1)=e_i^{j+1}$; and
\item by $\rho_i^j(t)$ the rotation of $e_i$ around $v_i$ during the linear morph transforming $e_i(0)$ in $e_i(t)$ -- note that $\rho_i^j(0)=0$, and $\rho_i^j(1)=\rho_i^j$.
\end{itemize}

Since a morph is a continuous transformation and since $|\rho_i^j| \geq \pi$, it follows that there exists a time instant $t_{\pi}$ with $0< t_{\pi}\leq 1$ such that $|\rho_i^j(t_{\pi})|=\pi$. We prove that there exists a time instant $t_r$ with $0< t_r\leq t_{\pi}$ in which $v_i(t_r)$ and $v_{i+1}(t_r)$ coincide, thus contradicting the assumption that morph $\langle \Gamma_{j},\Gamma_{j+1} \rangle$ is planar. Refer to Figure~\ref{fig:proof-lb-rotation}. 

\begin{figure}[htb]
\centering
\includegraphics[scale=0.75]{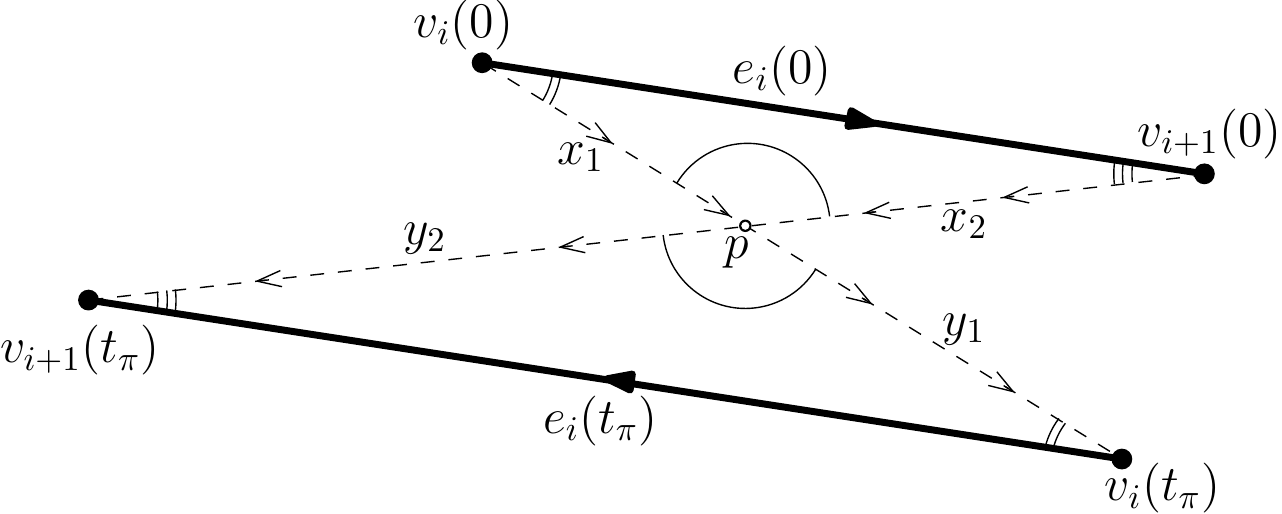}\label{fig:lb-rotation-proof-a}
\caption{Illustration for the proof of Lemma~\ref{le:lb-rotation}.}
\label{fig:proof-lb-rotation}
\end{figure}

Since $|\rho_i^j(t_{\pi})|=\pi$, it follows that $e_i(t_{\pi})$ is parallel to $e_i(0)$ and oriented in the opposite way. This easily implies that $t_r$ exists if $e_i(t_{\pi})$ and $e_i(0)$ are aligned. Otherwise, the straight-line segments $v_i(0) v_i(t_{\pi})$ and $v_{i+1}(0) v_{i+1}(t_{\pi})$ meet in a point $p$. Let $x_1$, $x_2$, $y_1$, and $y_2$ be the lengths of the straight-line segments $p v_i(0)$, $p v_{i+1}(0)$, $p v_i(t_{\pi})$, and $p v_{i+1}(t_{\pi})$, respectively. By the similarity of triangles $v_i(0) p v_{i+1}(0)$ and $v_i(t_{\pi}) p v_{i+1}(t_{\pi})$, we have $\frac{x_1}{y_1}=\frac{x_2}{y_2}$ and hence $\frac{x_1}{x_1+y_1}=\frac{x_2}{x_2+y_2}$. Thus, at time instant $t_r = \frac{x_1}{x_1+y_1}t_{\pi}$, we have that $v_i(t_r)$ and $v_{i+1}(t_r)$ coincide (in fact they both lie at $p$). This contradiction proves the lemma.
\end{proof}

For $j=1,\dots,x-1$, we denote by $M_j$ morph $\morph{\Gamma_1,\dots,\Gamma_{j+1}}$ -- note that $M_{x-1}=M$; also, for $i=1,\dots,n-1$, we define the \emph{total rotation} $\rho_i(M_j)$ of edge $e_i$ around $v_i$ during morph $M_j$ as $\rho_i(M_j)=\sum_{m=1}^{j}\rho_i^m$. Observe that the total rotation might be a value larger than $2\pi$ radians or smaller than $0$ radians; that is, the sum is not taken modulo $2\pi$. 

We will show in Lemma~\ref{le:lb-linear-total-rotation} that there exists an edge $e_i$, for some $1 \le i \le n-1$, whose total rotation during the entire morph is linear in the size of the path; that is, $\rho_i(M_{x-1})=\rho_i(M)\in \Omega(n)$. In order to do that, we first analyze the relationship between the total rotation of two consecutive edges of $P$.

\begin{lemma}\label{le:lb-diff-rot}
For each $j=1,\dots,x-1$ and for each $i=1,\dots,n-2$, we have that $|\rho_{i+1}(M_j)-\rho_{i}(M_j)|<\pi$.
\end{lemma}

\begin{proof}
Suppose, for a contradiction, that $|\rho_{i+1}(M_j)-\rho_{i}(M_j)|\geq \pi$ for some $1\leq j \leq x-1$ and $1\leq i \leq n-2$. Assume that $j$ is minimal under this hypothesis. 

Since each vertex moves continuously during $M_j$, there exists an intermediate drawing $\Gamma^*$ of $P$, occurring during morphing step $\morph{\Gamma_j,\Gamma_{j+1}}$, such that $|\rho_{i+1}(M^*)-\rho_{i}(M^*)| = \pi$, where $M^*=\morph{\Gamma_1,\dots,\Gamma_j,\Gamma^*}$ is the morph obtained by concatenating $M_{j-1}$ with the morphing step transforming $\Gamma_{j}$ into $\Gamma^*$.

Recall that in $\Gamma_1$ edges $e_i$ and $e_{i+1}$ lie on the same straight line and have the same orientation. Then, since $|\rho_{i+1}(M^*)-\rho_{i}(M^*)| = \pi$, in $\Gamma^*$ edges $e_i$ and $e_{i+1}$ are parallel and have opposite orientations. Also, since edges $e_i$ and $e_{i+1}$ share vertex $v_{i+1}$, they lie on the same line. This implies that such edges overlap, contradicting the hypothesis that $M^*$, $M_j$, and $M$ are planar.
\end{proof}

We now prove the key lemma for the lower bound.

\begin{lemma}\label{le:lb-linear-total-rotation}
There exists an index $i$ such that $|\rho_i(M)| \in \Omega(n)$.
\end{lemma}
\begin{proof}
Refer again to Figure~\ref{fig:lb-drawings}. For every $1\leq i\leq n-2$, edges $e_i$ and $e_{i+1}$ form an angle of $\pi$ radians in $\Gamma$, while they form an angle of $\frac{\pi}{3}$ radians in $\Lambda$. Hence, the total rotation of $e_{i+1}$ during the entire morph $M$ has to be larger than the one of $e_{i}$ by $\frac{2\pi}{3}$ radians (plus any multiple of $2\pi$); that is, $\rho_{i+1}(M)=\rho_i(M)+\frac{2\pi}{3} + 2z_i\pi$, for some $z_i\in \mathbb{Z}$.

In order to prove the lemma, it suffices to prove that $z_i=0$, for every $i=1,\dots,n-2$. Namely, in this case $\rho_{i+1}(M)=\rho_i(M)+\frac{2\pi}{3}$ for every $i=1,\dots,n-2$, and hence $\rho_{n-1}(M) = \rho_1(M) + \frac{2\pi}{3}(n-2)$. This implies $|\rho_{n-1}(M)-\rho_1(M)|\in \Omega(n)$, and thus $|\rho_1(M)|\in \Omega(n)$ or $|\rho_{n-1}(M)|\in \Omega(n)$.

Assume, for a contradiction, that $z_i\neq 0$, for some $1 \leq i \leq n-2$. If $z_i>0$, then $\rho_{i+1}(M)\geq \rho_i(M)+\frac{8\pi}{3}$; further, if $z_i<0$, then $\rho_{i+1}(M)\leq \rho_i(M)-\frac{4\pi}{3}$. Since each of these inequalities contradicts Lemma~\ref{le:lb-diff-rot}, the lemma follows.
\end{proof}

We are now ready to prove Theorem~\ref{th:lb-bound}. Namely, consider the two drawings $\Gamma$ and $\Lambda$ of path $P=(v_1,\dots,v_n)$ illustrated in Figure~\ref{fig:lb-drawings}. By Lemma~\ref{le:lb-linear-total-rotation}, there exists an edge $e_i$ of $P$, for some $1 \le i \le n-1$, such that $|\sum_{j=1}^{x-1}\rho_i^j| \in \Omega(n)$. Since, by Lemma~\ref{le:lb-rotation}, we have that $|\rho_i^j| < \pi$ for each $j=1,\dots,x-1$, it follows that $x \in \Omega(n)$. This concludes the proof of the theorem.

\section{Conclusion} \label{se:conclusions}

In this paper we have given an $O(n^3)$-time algorithm that takes as input two straight-line planar drawings $\Gamma_1$ and $\Gamma_2$ of the same $n$-vertex planar graph with the same embedding, and finds a morph consisting of $O(n)$ unidirectional morphing steps from $\Gamma_1$ to $\Gamma_2$ that preserves straight-line planarity.   The number of steps of our morph is asymptotically optimal in the worst case.


Our algorithm works under the real RAM model of computation.
However, we have not bounded the coordinate values used in our morphs,  
and it seems that they may require a super-logarithmic number of bits
 (though they can be described using a polynomial number
of arithmetic operations). 
Consequently, the intermediate drawings produced by our morph may have an exponential ratio of the distances between the closest and farthest pairs of vertices. We leave as an open problem to find a morph that uses a polynomial number of discrete morphing steps and uses only a logarithmic number of bits per coordinate. Barrera-Cruz et al.~\cite{Barrera-Cruz-Schnyder} made a first step in this direction by solving the case where the two drawings are Schnyder drawings.

\section*{Acknowledgements}
Part of this research was conducted in the framework of ESF project
10-EuroGIGA-OP-003 GraDR ``Graph Drawings and Representations'', of 
EU FP7 STREP Project ``Leone: From Global Measurements to Local
Management'', grant no. 317647, and of DFG grant Ka812/17-1.
Part of this research was supported by NSERC, the Natural Sciences and Engineering Research Council of Canada.  
Part of this research 
was supported by 
the Danish National Research Foundation grant DNRF84 through the Center for Massive Data Algorithmics (MADALGO).
Some of the work was done as part of an Algorithms Problem Session at the University of Waterloo, and we thank the other participants for helpful discussions. We thank Therese Biedl for a good suggestion about 3-connected graphs.

\bibliographystyle{spmpsci}      
\bibliography{journal}   

\end{document}